\documentclass[preprint,12pt,a4paper,twoside,fleqn,sort&compress]{elsarticle}
\usepackage{amssymb,latexsym,amscd,amsmath,amsthm}
\usepackage{pst-all}
\usepackage[varg]{pxfonts}
\usepackage{mathrsfs}
\usepackage{shuffle}
\usepackage{hyperref}
\usepackage{times}

\usepackage{color} 

\usepackage{mathrsfs}
\usepackage{graphicx}

\usepackage{helvet}
\usepackage{courier}

\usepackage{algorithm}
\usepackage[noend]{algorithmic}
\usepackage{verbatim}
\usepackage{flafter}
\newcommand{\rw}{\rightarrow}

\newcommand{\pp}{\Phi^{'}}
\newcommand{\ppp}{\Phi^{''}}

\theoremstyle{definition}
\newtheorem{definition}{Definition}[section]
\theoremstyle{plain}
\newtheorem{proposition}{Proposition}[section]
\newtheorem{theorem}{Theorem}[section]
\newtheorem*{thm-other}{Theorem}
\newtheorem{example}{Example}[section]

\newtheorem{lemma}{Lemma}
\theoremstyle{remark}
\newtheorem{remark}{Remark}

\begin{document}

\journal{arXiv}

\begin{frontmatter}


\title{Quantitative Computation Tree Logic Model Checking Based on Generalized Possibility Measures
\thanks{This work was partially supported by National Science
Foundation of China (Grant No: 11271237,61228305) and the Higher School Doctoral
Subject Foundation of Ministry of Education of China (Grant No:20130202110001).}}

\author {Yongming Li\corref{cor1}}
\ead{liyongm@snnu.edu.cn}

\author {Zhanyou Ma}

\address {College of Computer Science, Shaanxi Normal University, Xi'an, 710062, China}
\cortext[cor1]{Corresponding Author}
\begin{abstract}

We study generalized possibilistic computation tree logic model checking in this paper, which is an extension of possibilistic computation logic model checking introduced by Y.Li, Y.Li and Z.Ma \cite{li13}. The system is modeled by generalized possibilistic Kripke structures (GPKS, in short), and the verifying property is specified by a generalized possibilistic computation tree logic (GPoCTL, in short) formula. Based on generalized possibility measures and generalized necessity measures, the method of generalized possibilistic computation tree logic model checking is discussed, and the corresponding algorithm and its complexity are shown in detail. Furthermore, the comparison between PoCTL introduced in \cite{li13,Xue09} and GPoCTL is given. Finally, a thermostat example is given to illustrate the GPoCTL model-checking method.

\end{abstract}

\begin{keyword} Model checking; possibility theory; generalized possibilistic Kripke structure; generalized possibilistic computation tree logic; quantitative property.
\end{keyword}

\end{frontmatter}

\baselineskip 20pt

\section{Introduction}

Model checking \cite{EGP92} is a formal
verification technique consisting of three main steps: modeling the system, specifying the properties of the system (i.e., specification), and verifying whether the properties hold in the system using model-checking algorithms. Systems are usually represented using boolean state-transition models or Kripke structures. Properties of the system are often specified using temporal logics. The verification step gives a boolean answer: either true (the system satisfies the specification) or false with counterexample (the system violates the specification).


Boolean transition models are useful for the representation and verification of computation systems, such as hardware and software systems. However, boolean state-transition models are often inadequate for the representation of systems that are not purely computational but partly physical, such as hardware and software systems that interact with a physical environment and Cyber-Physical Systems (CPS). Many quantitative extensions of the state-transition model have been proposed for this purpose, such as models that embed state changes into time (\cite{Baier08}), models that assign probabilities (\cite{Baier08}) or possibilities (\cite{li12}) to state changes with uncertainties.

Furthermore, for the application to quantitative models and quantitative specifications, quantitative model-checking approaches have been proposed recently. Different approaches are applicable to different models types including timed (\cite{Baier08}), probabilistic and stochastic (\cite{huth97}), multi-valued (\cite{chechik012,chechik04,chechik06}), quality of service or soft constraints (\cite{lluch05}), discounted sources-restricted (\cite{dealfaro05}), possibilistic (\cite{li13}), etc, methods.

Although possibilistic CTL is more expressive than CTL, it is too restrictive (\cite{li13}). Some uncertainties, which can be described using possibility theory, still could not be handled directly using possibilistic computation tree logic model checking as noted in \cite{li13}, e.g. those systems modeled by possibilistic Kripke structures with vague label functions (see the definition of generalized possibilistic Kripke structures in Section 3 in this paper). To deal with uncertainties in possibility theory, more powerful quantitative model checking  is needed. For this purpose, we shall study quantitative model checking based on generalized possibilistic measures in this paper. Here, the models of systems are formalized as generalized possibilistic Kripke structures (GPKS). Compared with possibilistic Kripke structures (PKS), the initial distribution and state-transition distribution of GPKS have no normal condition restrictions, and the labeling function of GPKS is fuzzy and contains vague information. The specification is quantitative CTL which is called generalized possibilistic CTL (GPoCTL, in short), the interpretation of GPoCTL formula is also quantitative, even if the GPKS is also a PKS, and more possibilistic quantitative information is contained in GPoCTL compared with that in PoCTL, for example, the necessity measure is also introduced in the interpretation of GPoCTL formulae. The related model checking approach and its complexity are presented, and some comparisons are made between PoCTL and GPoCTL.

Since we can use fuzzy sets to represent multi-valued simulation, the techniques used in this paper have some similarities to those used in multi-valued cases (\cite{chechik012}). Of course, some essential differences exist. Indeed, possibilistic measures and necessity measures are used in GPoCTL. There is not any measure introduced for multi-valued cases. We give an illustrative example to show the approach proposed in this paper is efficient and reasonable. In fact, we expect that GPoCTL model checking will be used in the verification of expert systems and diagnosis of intelligent systems.

The content of this paper is arranged as follows. Section 2 gives some introduction of possibility theory, PoCTL and PKS defined in \cite{li12,li13}. Some possibility measures and necessity measures related to PKS and PoCTL are also studied. The necessity measures introduced in this section are new and not defined in \cite{li12,li13}. In Section 3 we give the notion of generalized possibilistic Kripke structures, the related generalized possibility measures induced by the generalized possibilistic Kripke structures. Section 4 introduces the notion of GPoCTL.
In Section 5, the GPoCTL model checking approach is discussed and the related algorithm is presented. Section 6 shows the relationship between GPoCTL and PoCTL. A thermostat example is given in Section 7.
The paper ends with a conclusion.

\section{Preliminaries}

In this section, we give some basic knowledge about the possibility theory, and recall the possibilistic computation tree logic (PoCTL, in short) introduced in \cite{li13}.

\subsection{Possibility theory}

Possibility theory is an uncertainty theory devoted to the handling of incomplete information and is an alternative to probability theory. It differs from the latter by the use of a pair of dual set-functions (possibility and necessity measures) instead of only one. This feature makes it easier to capture partial ignorance. Besides, it is not additive and makes sense on ordinal structures. Professor Lotfi Zadeh (\cite{Zadeh78}) first introduced possibility theory in 1978 as an extension of his theory of fuzzy sets and fuzzy logic. Didier Dubois and Henri Prade (\cite{dubois88,dubois06,Didier14,dubois11}) further contributed to its development.

For simplicity, assume that the universe of discourse $U$ is a nonempty set, and assume that all subsets are measurable. A possibility measure is a function $\Pi$ from the powerset $2^U$ to $[0, 1]$ such that:

(1) $\Pi(\emptyset)=0$, (2) $\Pi(U)=1$, and (3) $\Pi(\bigcup E_i)=\bigvee \Pi(E_i)$ for any subset family $\{E_i\}$ of the universe set $U$, where we use $\bigvee_{i\in I}a_i$ to denote the supremum or the least upper bound of the family of real numbers $\{a_i\}_{i\in I}$, dually, we use $\bigwedge_{i\in I}a_i$ to denote the infimum  or the largest lower bound of the family of real numbers $\{a_i\}_{i\in I}$.

If $\Pi$ only satisfies the conditions (1) and (3), then we call $\Pi$ a generalized possibility measure.

It follows that,the generalized possibility measure on a nonempty set is determined by its behavior on singletons:
\begin{equation}\label{eq:possibility distribution}
\Pi(E)=\bigvee_{x\in E} \Pi(\{x\}).
\end{equation}
The function $\pi: U\longrightarrow [0,1]$ defined by $\pi(x)=\Pi(\{x\})$ is called the possibility distribution of $\Pi$, and the measure $\Pi$ is unique defined by Eq.(\ref{eq:possibility distribution}), i.e., $\Pi$ is unique defined by the possibility distribution $\pi$.

Whereas probability theory uses a single number, the probability, to describe how likely an event is to occur, possibility theory uses two concepts, the possibility and the necessity of the event. For any set $E$, the necessity measure $N$ is defined by,
\begin{equation}\label{eq:necessity measure}
N(E)=1-\Pi(U-E).
\end{equation}
A necessity measure is a function $N$ from the powerset $2^U$ to $[0, 1]$ such that:

(1) $N(\emptyset)=0$, (2) $N(U)=1$, and (3) $N(\bigcap E_i)=\bigwedge N(E_i)$ for any subset family $\{E_i\}$ of the universe set $U$.

If $N$ only satisfies the conditions (2) and (3), then we call $N$ a generalized necessity measure.

{\bf It follows that $\Pi(E)+N(U-E)=1$, and $N$ is the dual of $\Pi$ and vise versa. In general, $\Pi$ and $N$ are not self-dual, this is contrary to probability measure, which is self-dual. As a result, we need both possibility measure and necessity measure to treat uncertainty in the theory of possibility.}

There are four cases that can be interpreted as follows:
(1) $N(E)=1$ means that $E$ is necessary. $E$ is certainly true. It implies that $\Pi(E)=1$.
(2) $\Pi(E)=0$ means that $E$ is impossible. $E$  is certainly false. It implies that $N(E)=0$.
(3) $\Pi(E)=1$  means that $E$ is possible. It would not be surprised at all if $E$ occurs. It leaves $N(E)$  unconstrained.
(4) $N(E)=0$ means that $E$ is unnecessary. It would not be surprised at all if $E$ does not occur. It leaves $\Pi(E)$ unconstrained.

We shall use possibility measures and necessity measures in the possibilistic computation tree logic model checking in this paper.

\subsection{Possibilistic Kripke structures}
Transition systems or Kripke structures are key representations for model checking. Corresponding to possibilistic model checking, we have the notion of possibilistic Kripke structures, which is defined as follows.

\begin{definition} \cite{li12}\label{de:pkripke}
A possibilistic Kripke structure is a tuple $M=(S,P,I,AP,L)$, where

(1) $S$  is a countable, nonempty set of states;

(2) $P:S\times S\longrightarrow [0,1]$ is the transition possibility distribution such that for all states $s$, $\bigvee\limits_{s^{'}\in S}P(s,s^{'})=1$ ;

(3) $I:S\longrightarrow[0,1]$ is the initial distribution, such that $\bigvee\limits_{s\in S}I(s)=1$ ;

(4) $AP$ is a set of atomic propositions;

(5) $L:S\longrightarrow 2^{AP}$  is a labeling function that labels a state $s$ with those atomic propositions in $AP$ that are supposed to hold in $s$.

Furthermore, if the set  $S$  and  $AP$ are finite sets, then $M=(S,P,I,AP,L)$ is called a
finite possibilistic Kripke structure.
\end{definition}

The states $s$ with $I(s)>0$ are considered as the initial states. For state $s$ and $T\subseteq S$, let $P(s,T)$ denote the possibility of moving from $s$ to some state $t\in T$ in a single step, that is,
\begin{equation*}
P(s,T)=\vee_{t\in T}P(s,t).
\end{equation*}

Paths in possibilistic Kripke structure $M$ are infinite paths in the underlying digraph. They are
defined as infinite state sequences $\pi=s_{0}s_{1}s_{2}\cdots\in S^{w}$  such that  $P(s_{i},s_{i+1})>0$ for all $i\in I$.
Let $Paths(M)$ denote the set of all paths in $M$, and $Paths_{fin}(M)$ denote the set of finite path
fragments $s_{0}s_{1}\cdots s_{n}$ where $n\geq 0$ and $P(s_{i},s_{i+1})>0$ for $0\leq i\leq n$. Let $Paths_M(s)$ ($Paths(s)$ if $M$ is understood) denote the set of all
paths in $M$ that start in state $s$. Similarly, $Paths_{M-fin}(s)$ ($Paths_{fin}(s)$ if $M$ is understood) denotes the set of finite path fragments $s_{0}s_{1}\cdots s_{n}$ such that $s_{0}=s$.
The set of direct successors (called  $Post$ ) and direct predecessors (named  $Pre$ ) are defined
as follows:
\begin{equation*}
Post(s)=\{s'\in S\mid P(s,s')> 0\};~~ Pre(s)=\{s' \in S\mid P(s',s)> 0\}.
\end{equation*}

Given a possibilistic Kripke structure $M$, the cylinder set of $\hat{\pi}=s_0\cdots s_n\in Paths_{fin}(M)$ is defined as (\cite{Baier08}) $$Cyl(\hat{\pi})=\{\pi\in Paths(M) | \hat{\pi}\in Pref(\pi)\},$$ where $Pref(\pi)=\{\pi^{\prime} | \pi^{\prime}$ is a finite prefix of $\pi\}$. Then as shown in \cite{li12},  $\Omega=2^{Paths(M)}$ is the algebra generated by $\{Cyl(\hat{\pi})\mid\hat{\pi}\in Paths_{fin}(M)\}$ on $Paths(M)$. That is to say, $\Omega=2^{Paths(M)}$ is the unique subalgebra of $2^{Paths(M)}$ which is closed under arbitrary unions and arbitrary intersections containing $\{Cyl(\hat{\pi}) | \hat{\pi}\in Pref(\pi)\}$.

\begin{definition}\label{def:possibility measure} \cite{li12} For a possibilistic Kripke structure $M$, a function $Po^M: Paths(M)\rightarrow [0,1]$ is defined as follows:
\begin{equation}\label{eq:possibility measure-path}
Po^M(\pi)=I(s_{0})\wedge\bigwedge\limits_{i=0}^\infty P(s_{i},s_{i+1})
\end{equation}
for any $\pi=s_{0}s_{1} \cdots, \pi\in Paths(M).$
Furthermore, we define
\begin{equation}\label{eq:possibility measure}
Po^M(E)=\vee\{Po^M(\pi)\mid\pi\in E\}
\end{equation}
for any $E\subseteq Paths(M)$, then, we have a well-defined function $$Po^M:2^{Paths(M)}\longrightarrow [0,1],$$ $Po^M$ is called the possibility measure over $\Omega=2^{Paths(M)}$ as it satisfies the definition of possibility measure. If $M$ is clear from the context, then $M$ is omitted and we simply write $Po$ for $Po^M$.

\end{definition}

For the above possibility measure $Po$ over $2^{Paths(M)}$, the corresponding necessity measure, write as $Ne$, is defined as follows,

$Ne(E)=1-Po(\overline{E})$,

\noindent where $\overline{E}$ denotes the complement of the subset $E$, i.e., $\overline{E}=Paths(M)-E$.

\subsection{Possibilistic computation tree logic}
\begin{definition} \cite{Xue09} (Syntax of PoCTL) {\sl PoCTL state formulae} over the set $AP$ of atomic propositions are formed according to the following grammar:
\begin{center}
$\Phi ::= true\mid a \mid\Phi_{1} \wedge \Phi_{2}\mid \neg \Phi\mid Po_{J}(\varphi)$
\end{center}
where $a\in AP$, $\varphi$ is a PoCTL path formula and $J\subseteq [0,1]$ is an interval with rational bounds.

{\sl PoCTL path formulae} are formed according to the following grammar:

\begin{center}
$\varphi::=\bigcirc \Phi \mid \Phi_{1}\sqcup \Phi_{2}\mid \Phi_{1}\sqcup^{\leq n} \Phi_{2} $
\end{center}
where $ \Phi$, $ \Phi_{1}$, and $ \Phi_{2}$ are state formulae and $n\in\mathbb{N}$.
\end{definition}

\begin{definition}\label{def:semantics of PoCTL}\cite{Xue09} (Semantics of PoCTL) Let $a\in AP$ be an atomic proposition, $M=(S,P,I,AP,L)$ be a possibilistic Kripke structure, state $s\in S$, $\Phi$, $\Psi$ be PoCTL state formulae, and $\varphi$ be a PoCTL path formula. {\sl The satisfaction relation $\models$} is defined {\sl for state formulae} by
\begin{eqnarray*}
s\models a  & {\rm iff} \ a\in L(s);\\
s\models\neg\Phi & {\rm iff}\ s\not\models\Phi; \\
s\models\Phi\wedge\Psi  & {\rm iff} \ s\models\Phi \ {\rm and}\ s\models\Psi;\\
s\models Po_{J}(\varphi)  & {\rm iff} \ Po(s\models \varphi)\in J, \ {\rm where}\ Po(s\models \varphi)=Po^{M_s}(\{\pi | \pi\in Paths(s), \pi\models \varphi\}).
\end{eqnarray*}
\noindent where $M_s$ results from $M$ by letting $s$ be the unique initial state. Formally, for $M=(S,P,I,AP,L)$ and state $s$,  $M_s$ is defined by $M_s=(S,P,s,AP,L)$ , where $s$ denotes an initial distribution with only one initial state $s$.

For path $\pi$, {\sl the satisfaction relation $\models$ for path formulae} is defined by
\begin{eqnarray*}
\pi\models\bigcirc\Phi  & {\rm iff} \ \pi[1]\models\Phi;\\
\pi\models\Phi\sqcup\Psi  & {\rm iff} \  \exists k\geq0,\pi[k]\models\Psi
\ {\rm and}\
 \pi[i]\models\Phi {\rm \ for\ all} \ 0\leq i\leq k-1;\\
\pi \models \Phi\sqcup^{\leq n}\Psi & {\rm iff} \ \exists 0\leq k\leq n, (\pi[k]\models \Psi\wedge(\forall 0\leq i< k),\pi[i]\models \Phi)).
\end{eqnarray*}
where if $\pi=s_0s_1s_2\cdots$, then $\pi[k]=s_k$ for any $k\geq 0$.
\end{definition}

In particular, the path formulae $\lozenge\Phi$ (``eventually'') and $\square\Phi$ (``always'') have the semantics
$$\pi=s_0s_1\cdots\models \lozenge\Phi {\rm \ iff}\ s_j\models \Phi {\rm \ for\ some\ } j\geq 0,$$
$$\pi=s_0s_1\cdots\models \square\Phi {\rm \ iff}\ s_j\models \Phi {\rm \ for\ all\ } j\geq 0.$$ Alternatively, $\lozenge\Phi=true\sqcup \Phi$.

The intend meaning of the formula $Po(s\models \varphi)$ is the possibility measure of those paths starting at state $s$ satisfy the path formula $\varphi$ for any state $s$, that is,

$Po(s\models \varphi)=Po^{M_s}(\{\pi | \pi\in Paths(s), \pi\models \varphi\})$.

Let us see how the necessity measure can be defined in the interpretation of the PoCTL formulae.

Since $\pi\not\models \bigcirc \Phi$ iff $\pi[1]\not\models \Phi$ iff $\pi[1]\models \neg\Phi$ iff $\pi\models \bigcirc\neg\Phi$, it follows that

$\{\pi | \pi\in paths(s), \pi\not\models \bigcirc \Phi\}=\{\pi | \pi\in paths(s), \pi\models \bigcirc\neg\Phi\}$,

\noindent then we have
\begin{eqnarray*}
\{\pi | \pi\in paths(s), \pi\models \bigcirc\Phi\}&=&\overline{\{\pi | \pi\in paths(s), \pi\not\models \bigcirc \Phi\}}\\
&=&\overline{\{\pi | \pi\in paths(s), \pi\models \bigcirc\neg\Phi\}}.
  \end{eqnarray*}
Hence,
\begin{eqnarray*}
Ne(s\models \bigcirc \Phi)&=& Ne^{M_s}(\{\pi | \pi\in paths(s), \pi\models \bigcirc\Phi\}) \\
&=&1-Po^{M_s}(\{\pi | \pi\in paths(s), \pi\models \bigcirc\neg\Phi\})\\
&=&1-Po(s\models \bigcirc\neg\Phi).
  \end{eqnarray*}
Similarly, we have the following equations,

$Ne(s\models \Phi \sqcup\Psi)=(1-Po(s\models \neg \Psi\sqcup (\neg\Phi\wedge \neg\Psi)))\wedge (1-Po(s\models \square\neg \Phi))$,

$Ne(s\models\Phi \sqcup^{\leq n}\Psi)=(1- Po(s\models\neg \Psi\sqcup^{\leq n} (\neg\Phi\wedge \neg\Psi)))\wedge (1- Po(s\models\square^{\leq n}\neg \Phi))$,

$Ne(s\models \square\Phi)=1-Po(s\models \lozenge\neg \Psi)$,

$Ne(s\models \lozenge\Phi)=1-Po(s\models \square\neg \Psi)$.

If we write a PoCTL state formula $Ne_J(\varphi)$ for a path formula $\varphi$, which have the semantics

$s\models Ne_J(\varphi)$ iff $Ne^{M_s}(\{\pi\in Paths(s) | \pi\models \varphi\})\in J$

\noindent for any PKS $M$, then we have the following presentation of $Ne_J(\varphi)$, where for interval $J=[u,v], (u,v], [u,v), (u,v)$, $DJ=[1-v,1-u], [1-v,1-u), (1-v,1-u], (1-v,1-u)$:
\begin{equation}\label{eq:necessity 1}
Ne_J(\bigcirc \Phi)=\neg Po_{DJ}(\bigcirc\neg\Phi);
\end{equation}
\vspace{-0.6cm}
\begin{equation}\label{eq:necessity 2}
Ne_J(\Phi \sqcup\Psi)=\neg Po_{DJ}(\neg \Psi\sqcup (\neg\Phi\wedge \neg\Psi))\wedge \neg Po_{DJ}(\square\neg \Phi);
\end{equation}
\vspace{-0.6cm}
\begin{equation}\label{eq:necessity 3}
Ne_J(\Phi \sqcup^{\leq n}\Psi)=\neg Po_{DJ}(\neg \Psi\sqcup^{\leq n} (\neg\Phi\wedge \neg\Psi))\wedge \neg Po_{DJ}(\square^{\leq n}\neg \Phi);
\end{equation}
\vspace{-0.6cm}
\begin{equation}\label{eq:necessity 4}
Ne_J(\square\Phi)=\neg Po_{DJ}(\lozenge\neg \Psi);
\end{equation}
\vspace{-0.6cm}
\begin{equation}\label{eq:necessity 5}
Ne_J(\lozenge\Phi)=\neg Po_{DJ}(\square\neg \Psi).
\end{equation}

The above equalities are also the sources that we define the GPoCTL formula $Ne(\bigcirc \Phi), Ne(\Phi \sqcup\Psi), Ne(\Phi \sqcup^{\leq n}\Psi), Ne(\square\Phi)$ and $Ne(\lozenge\Phi)$ in Section 5.

\section{Generalized possibilistic Kripke structures}

 In this section, we extend the notion of PKS and introduce the notion of generalized possibilistic Kripke structures, which is defined as follows.

\begin{definition} \label{de:pkripke}
A generalized possibilistic Kripke structure (GPKS, in short) is a tuple $M=(S,P,I,AP,L)$, where

(1) $S$  is a countable, nonempty set of states;

(2) $P:S\times S\longrightarrow [0,1]$ is a function, called possibilistic transition distribution function;

(3) $I:S\longrightarrow [0,1]$ is a function, called possibilistic initial distribution function;

(4) $AP$ is a set of atomic propositions;

(5) $L:S\times AP\longrightarrow [0,1]$ is a possibilistic labeling function, which can be viewed as function mapping a state $s$ to the fuzzy set of atomic propositions which are possible in the state $s$, i.e., $L(s,a)$ denotes the possibility or truth value of atomic proposition $a$ that is supposed to hold in $s$.

Furthermore, if the set  $S$  and  $AP$ are finite sets, then $M=(S,P,I,AP,L)$ is called a
finite generalized possibilistic Kripke structure.
\end{definition}

\begin{remark} (1) In Definition \ref{de:pkripke}, if we require the transition possibility distribution and initial distribution to be {\sl normal}, i.e., $\vee_{s'\in S}P(s,s')=1$ and $\vee_{s\in S}I(s)=1$, and the labeling function $L$ is also crisp, i.e., $L: S\times AP\longrightarrow \{0,1\}$. Then we obtain the notion of possibilistic Kripke structure (PKS, in short). In this case, we also say that $M$ is normal. This is one of the reasons why we call the structure defined in Definition \ref{de:pkripke} generalized possibilistic Kripke structure.

(2) The possibilistic transition function $P: S\times S\longrightarrow [0,1]$ can also be represented by a fuzzy matrix. For convenience, this fuzzy matrix is also written as $P$, i.e., $$P=(P(s,t))_{s,t\in S},$$ $P$ is also called the (fuzzy) transition matrix of $M$. In \cite{li12}, we also use the symbol $A$ to represent a transition matrix. For the fuzzy matrix $P$, its transitive closure is denoted by $P^+$. When $S$ is finite, and if $S$ has $N$ elements, i.e., $N=|S|$, then $P^+=P\vee P^2\vee\cdots\vee P^N$ \cite{li05}, where $P^{k+1}=P^k\circ P$ for any positive integer number $k$. Here, we use the symbol $\circ$ to represent the max-min composition operation of fuzzy matrixes. Recall that the max-min composition operation
of fuzzy matrixes is similar to ordinary matrix multiplication operation, just let ordinary multiplication and addition operations of real numbers be replaced by minimum and maximum operations of real numbers (\cite{Zadeh78}). For a fuzzy matrix $P$, the reflective and transitive closure of $P$, denoted by $P^{\ast}$, is defined by $P^{\ast}=P^0\vee P^+$, where $P^0$ denote the identity matrix.

For a generalized possibilistic Kripke structure $M=(S,P,I,AP,L)$, using $P^+$ and $P^{\ast}$, we can get two generalized possibilistic Kripke structures $M^+=(S,P^+,I,AP,L)$ and $M^{\ast}=(S,P^{\ast},I,AP,L)$.

(3) A closely related notion is given by (discrete-time) fuzzy Markov chains \cite{kruse87} or (discrete-time) possibilistic Markov chains (\cite{dubois94}) or possibilistic Markov processes (\cite{janssen96}) which are used to model certain fuzzy systems. The only difference between possibilistic Kripke structures and fuzzy (or possibilistic) Markov chains lies in that there is no labeling function in the definition of fuzzy (or possibilistic) Markov chains. In \cite{dubois94}, possibilistic Markov chains are used to model the evolution of the updating problem in a knowledge base that describes the state of an evolving system. Uncertainty comes from incomplete knowledge about the knowledge base, ``one may only have some idea about what is/are the most plausible state(s) of the system, among possible one''(\cite{dubois94}). This type of incomplete knowledge was described in terms of possibility distribution in \cite{dubois94}, the degree of transition possibility distribution denotes the plausible degree of the next state. This provide us one kind of view on the justification of degree and transition of possibilistic Kripke structures.
\end{remark}


\begin{example}\label{ex:running example}
{\rm Let us give some running examples of GPKSs, where states are represented by nodes and transitions by labeled edges.
State names are depicted inside the ovals. Initial states are indicated by having an incoming arrow without source.

 (1) Fig.1 shows a GPKS with fuzzy $P$ and $L$;

 (2) Fig.2 gives a GPKS with crisp $P$ and fuzzy $L$;

 (3) Fig.3 is a PKS;

 (4) Fig.4 presents a GPKS with non-normal fuzzy $P$ and crisp $L$.}

\end{example}

\begin{figure}[ht]
\begin{center}
\includegraphics[scale=0.8]{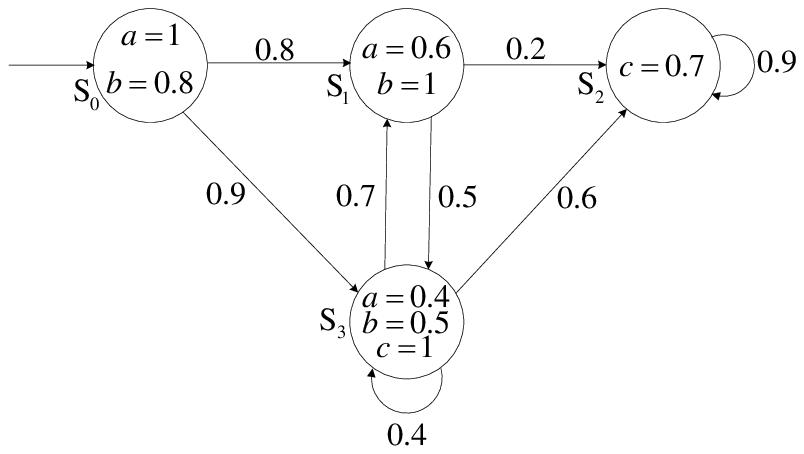}
\center{Fig.1.} A GPKS with fuzzy $P$ and $L$.
\vspace{-0.3cm}
\end{center}
\end{figure}

\begin{figure}[ht]
\begin{center}
\includegraphics[scale=0.8]{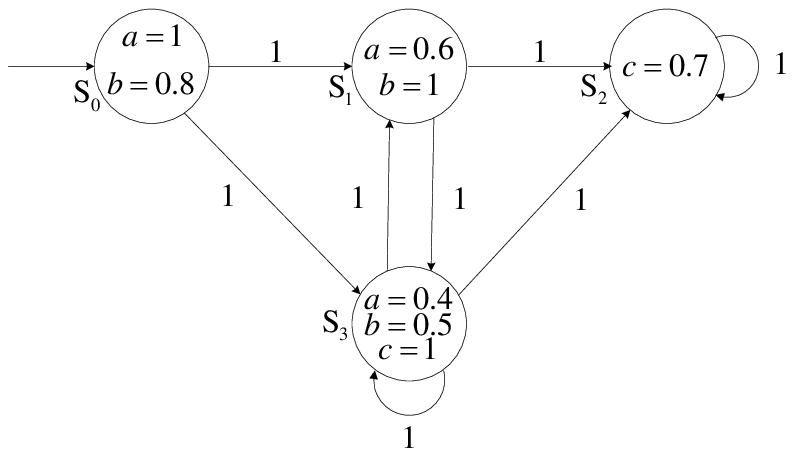}
\center{Fig.2.} A GPKS with crisp $P$ and fuzzy $L$.
\vspace{-0.3cm}
\end{center}
\end{figure}

\begin{figure}[ht]
\begin{center}
\includegraphics[scale=0.8]{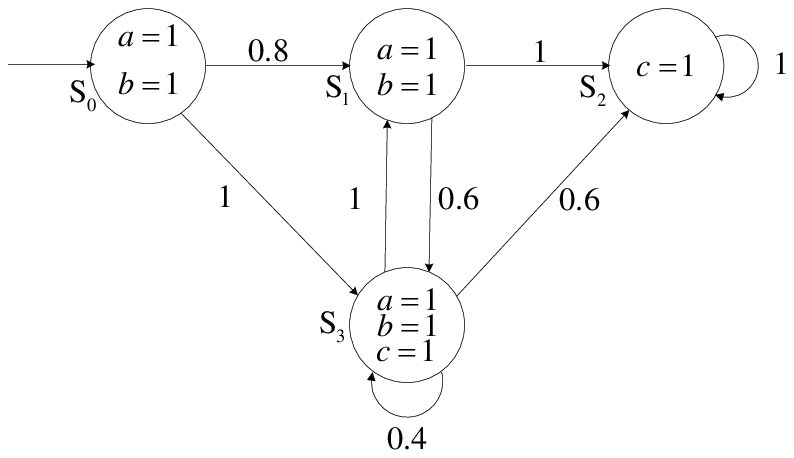}
\center{Fig.3.} A PKS.
\vspace{-0.3cm}
\end{center}
\end{figure}

\begin{figure}[ht]
\begin{center}
\includegraphics[scale=0.8]{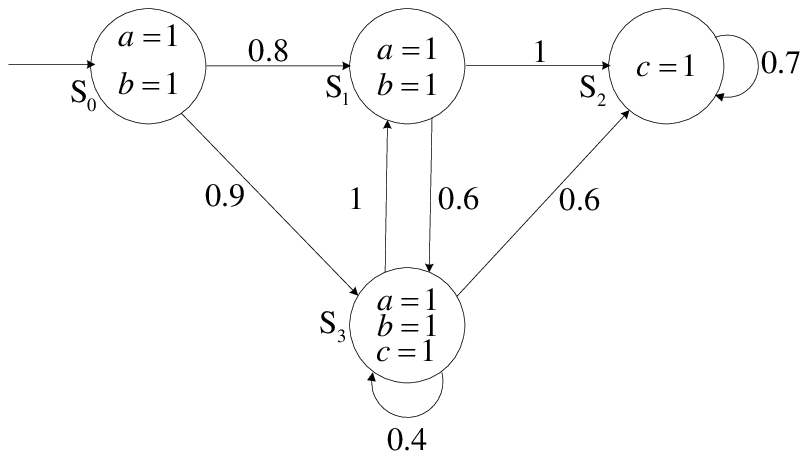}
\center{Fig.4.}A GPKS with non-normal fuzzy $P$ and crisp $L$.
\vspace{-0.3cm}
\end{center}
\end{figure}

The similar notions and notations used for PKS are also applicable for GPKS.

\begin{definition}\label{def:possibility measure} (cf.\cite{li12}) For a generalized possibilistic Kripke structure $M$, a function $Po^M: Paths(M)\rightarrow [0,1]$ is defined as follows:
\begin{equation}\label{eq:possibility measure-path}
Po^M(\pi)=I(s_{0})\wedge\bigwedge\limits_{i=0}^\infty P(s_{i},s_{i+1})
\end{equation}
for any $\pi=s_{0}s_{1} \cdots \in Paths(M).$
Furthermore, we define
\begin{equation}\label{eq:possibility measure}
Po^M(E)=\vee\{Po^M(\pi)\mid\pi\in E\}
\end{equation}
for any $E\subseteq Paths(M)$, then, we have a well-defined function $$Po^M:2^{Paths(M)}\longrightarrow [0,1],$$ $Po^M$ is called the generalized possibility measure over $\Omega=2^{Paths(M)}$ as it has the properties stated in Theorem \ref{th:possibility measure}. If $M$ is clear from the context, then $M$ is omitted and we simply write $Po$ for $Po^M$.

\end{definition}

For a generalized Kripke structure $M$, let us define a function $r_P:S\longrightarrow [0,1]$ as follows, which denotes the largest possibility of the paths in $M$ initialized at the state $s$,
\begin{equation}\label{eq:r-fucntion}
r_P(s)=\bigvee\{P(s,s_1)\wedge P(s_1,s_2)\wedge\cdots | s_1,s_2,\cdots \in S\}.
\end{equation}

The role of the function $r_P$ is stated in Theorem \ref{th:possibility measure}.

How to calculate $r_P$? The following proposition gives an answer.

\begin{proposition}\label{pr:r-function}
For a finite generalized Kripke structure $M$, and a state $s$ in $M$, we have
\begin{equation}\label{eq:r-fucntion}
r_P(s)=\bigvee\{P^+(s,t)\wedge P^+(t,t) | t \in S\}.
\end{equation}
In the matrix notation we have,
\begin{equation}\label{eq:r-fucntion-matrix}
r_P=P^+\circ D,
\end{equation}
\noindent where $D=(P^+(t,t))_{t\in S}.$

In particular, $P$ is normal iff $r_P(s)=1$ for any state $s$.
\end{proposition}

\begin{proof}
Since $S$ is finite, the image set of $P$ is also finite. Observing that the meet operation $\wedge$ does not generate new elements, it follows that the set $\{P(s,s_1)\wedge P(s_1,s_2)\wedge\cdots | s_1,s_2,\cdots \in S\}$ is also finite. Therefore, there exists a sequence $s_1,s_2,\cdots \in S$ such that $r_P(s)=P(s,s_1)\wedge P(s_1,s_2)\wedge\cdots$. Since $S$ is finite, there exist $t\in S$ and $i<j$ such that $s_i=s_j=t$. In this case, $P(s,s_1)\wedge P(s_1,s_2)\wedge\cdots=(P(s,s_1)\wedge\cdots\wedge P(s_{i-1},t))\wedge (P(t,s_{i+1})\wedge\cdots\wedge P(s_{j-1},t))\wedge\cdots\leq (P(s,s_1)\wedge\cdots\wedge P(s_{i-1},t))\wedge (P(t,s_{i+1})\wedge\cdots\wedge P(s_{j-1},t))\leq P^+(s,t)\wedge P^+(t,t)$. Hence, $r_P(s)\leq \bigvee\{P^+(s,t)\wedge P^+(t,t) | t \in S\}$.

Conversely, for any $t\in S$, by the definition of $P^+$, it follows that there exists $s_1,\cdots,s_i=t\in S$ and $s_{i+1},\cdots,s_j$ such that $P^+(s,t)=P(s,s_1)\wedge\cdots\wedge P(s_{i-1},t)$ and $P^+(t,t)=P(t,s_{i+1})\wedge\cdots \wedge P(s_j,t)$. Let $\pi=ss_1\cdots s_{i-1}t(s_{i+1}\cdots s_jt)^{\omega}$, then $P^+(s,t)\wedge P^+(t,t)=P(s,\pi[1])\wedge P(\pi[1],\pi[2])\wedge\cdots$. Hence, $P^+(s,t)\wedge P^+(t,t)\leq r_P(s)$, and thus $\bigvee\{P^+(s,t)\wedge P^+(t,t) | t \in S\}\leq r_P(s)$.

Therefore, $r_P(s)=\bigvee\{P^+(s,t)\wedge P^+(t,t) | t \in S\}.$

Furthermore, if $P$ is normal, i.e., $\bigvee_{t^{\prime}\in S}P(t,t^{\prime})=1$ for any $t\in S$, since $S$ is finite, it follows that there exists $t^{\prime}\in S$ such that $P(t,t^{\prime})=1$ for any $t\in S$. By this observation, from the state $s$, we can choose a sequence of states $s_1,s_2, \cdots$ such that $P(s_i,s_{i+1})=1$ for any $i\geq 0$. This sequence guarantees that $r_P(s)=1$ for any state $s$. Conversely, if $r_P(s)=1$ for any state $s$, then it is obvious that $P$ is normal.
\end{proof}

\begin{theorem}\label{th:possibility on cyl} Let $M$ be a  finite generalized possibilistic Kripke structure. Then the possibility measure of the cylinder sets is given by $Po(Cyl(s_{0}\cdots s_{n}))=I(s_{0})\wedge \bigwedge\limits_{i=0}^{n-1} P(s_{i},s_{i+1})\wedge r_P(s_n)$ when $n>0$ and $Po(Cyl(s_{0}))=I(s_{0})\wedge r_P(s_0)$.
\end{theorem}

\begin{proof} As $Cyl(s_0\cdots s_n)=\cup\{\pi\in S^{\omega}|s_0\cdots s_n\in Pref(\pi)\}$, we have
\begin{eqnarray*}
&&Po(Cyl(s_0\cdots s_n))\\
&=&\bigvee\{Po(\pi)|s_0\cdots s_n\in Pref(\pi)\}\\
&=&\bigvee\{I(s_0)\wedge\bigwedge_{i=0}^{\infty} P(s_i,s_{i+1}) | s_{n+1},  \cdots\in S\}\\
&=&\{I(s_0)\wedge\bigwedge_{i=0}^{n-1} P(s_i, s_{i+1})\}\wedge\bigvee\{\bigwedge_{i=n}^\infty P(s_i,s_{i+1})| s_i\in S,  i> n\}\\
&=& I(s_{0})\wedge \bigwedge\limits_{i=0}^{n-1} P(s_{i},s_{i+1})\wedge r_P(s_n).
  \end{eqnarray*}
\end{proof}

\begin{theorem}\label{th:possibility measure}  $Po$ is a generalized possibility measure (\cite{grabisch00}) on $\Omega=2^{Paths(M)}$, i.e., $Po$ satisfies the following conditions:

 (1) $Po(\varnothing)=0$, $Po(\Omega)=\bigvee_{s\in S} I(s)\wedge r_P(s)$;

 (2) $Po(\bigcup\limits_{i\in I}A_{i})=\bigvee\limits_{i\in I}Po(A_{i})$ for any $A_{i}\in\Omega$, $i\in I$.

\end{theorem}

The proof is direct.

For the above generalized possibility measure $Po$, the related generalized necessity is also denoted by $Ne$, i.e., $Ne(E)=1-Po(\overline{E})$ for any subset $E$ of $Paths(M)$.

\section{Generalized possibilistic CTL}
We shall give the temporal logic used for the specifications in this section. We shall introduce a new kind of quantitative temporal logics, which is called generalized possibilistic CTL.

\begin{definition}  (Syntax of GPoCTL) Generalized possibilistic CTL (GPoCTL, in short) state formulae over the set $AP$ of atomic propositions are formed according to the following grammar:
\begin{center}
$\Phi ::= true\mid  a \mid\Phi_{1} \wedge \Phi_{2}\mid \neg \Phi\mid Po(\varphi)$
\end{center}
where $a\in AP$, $\varphi$ is a PoCTL path formula.

PoCTL path formulae are formed according to the following grammar:

\begin{center}
$\varphi::=\bigcirc \Phi \mid \Phi_{1}\sqcup \Phi_{2}\mid \Phi_{1}\sqcup^{\leq n} \Phi_{2} |\square\Phi$
\end{center}
where $\Phi$, $ \Phi_{1}$, and $ \Phi_{2}$ are state formulae and $n\in\mathbb{N}$.
\end{definition}

Using the connectives $\wedge$ and $\neg$, other connectives, such as disjunction $\vee$, implication $\rw$, equivalence $\leftrightarrow$ can be derived as usual,

$\Phi_1\vee \Phi_2=\neg(\neg\Phi_1\wedge\neg\Phi_2)$,

$\Phi_1\rw \Phi_2=\neg\Phi_1\vee \Phi_2$,

$\Phi_1\leftrightarrow \Phi_2=(\Phi_1\rw\Phi_2)\wedge (\Phi_2\rw\Phi_1)$.

\begin{definition}\label{def:semantics of PoCTL}(Semantics of PoCTL) Let $a\in AP$ be an atomic proposition, $M=(S,P,I,AP,L)$ be a possibilistic Kripke structure, $s\in S$ be a state, $\Phi$, $\Psi$ be PoCTL state formulae, and $\varphi$ be a PoCTL path formula. For state formula $\Phi$, its semantics is a fuzzy set $||\Phi||: S\rw [0,1]$, which is defined recursively as follows, for any $s\in S$,

\vspace{-0.6cm}
\begin{equation}\label{eq:true}
||true||(s)=1;
\end{equation}
\vspace{-0.7cm}
\begin{equation}\label{eq:a}
||a||(s)=L(s,a);
\end{equation}
\vspace{-0.6cm}
\begin{equation}\label{eq:wedge}
||\Phi\wedge\Psi||(s)=||\Phi||(s)\wedge ||\Psi||(s);\\
\end{equation}
\vspace{-0.6cm}
\begin{equation}\label{eq:negation}
||\neg\Phi||(s)=1-||\Phi||(s); \\
\end{equation}
\vspace{-0.6cm}
\begin{equation}\label{eq:po}
||Po(\varphi)||(s)=Po(s\models \varphi).\\
\end{equation}

For a path formula $\varphi$, its semantics is a fuzzy set $||\varphi||: Paths(M)\rw [0,1]$, which is defined recursively for $\pi\in Paths(M)$ as follows,
\begin{equation*}\label{eq:bigcirc}
||\bigcirc\Phi||(\pi)=P(\pi[0],\pi[1])\wedge||\Phi||(\pi[1]);
\end{equation*}
\vspace{-0.6cm}
\begin{eqnarray*}\label{eq:sqcup}
||\Phi\sqcup\Psi||(\pi)&=||\Psi||(\pi[0])\vee\bigvee_{j>0}((||\Phi||(\pi[0])\wedge \bigwedge_{k<j}P(\pi[k-1],\pi[k])\\
&\wedge||\Phi||(\pi[k]))\wedge P(\pi[j-1],\pi[j])\wedge||\Psi||(\pi[j])));
 \end{eqnarray*}
\vspace{-0.6cm}
\begin{eqnarray*}\label{eq:restricted sqcup}
||\Phi\sqcup^{\leq n}\Psi||(\pi)&=||\Psi||(\pi[0])\vee\bigvee_{0<j\leq n}((||\Phi||(\pi[0])\wedge
 \bigwedge_{k<j}P(\pi[k-1],\pi[k])\\
 &\wedge||\Phi||(\pi[k]))\wedge P(\pi[j-1],\pi[j])\wedge||\Psi||(\pi[j]))).
\end{eqnarray*}
\vspace{-0.6cm}
\begin{equation*}\label{eq:square}
||\square \Phi||(\pi)=\bigwedge_{i=0}^{\infty}\bigwedge_{j=0}^{i-1} P(\pi([j]),\pi([j+1]))\wedge ||\Phi||(\pi([i])).
\end{equation*}

$Po(s\models \varphi)$ is defined as follows
\begin{equation}\label{eq:possibility}
Po(s\models \varphi)=\bigvee_{\pi\in Paths(s)}Po^{M_s}(\pi)\wedge ||\varphi||(\pi).
\end{equation}
Intuitively, $Po(s\models \varphi)$  denotes the largest possibility of the paths strating at $s$ satisfying the formula $\varphi$.

\end{definition}

Path formula $\lozenge\Phi$ (``eventually'') defined by $\lozenge\Phi=true\sqcup \Phi$ has the semantics
\begin{equation}\label{eq:eventually}
||\lozenge\Phi||(\pi)=\bigvee_{j=0}^{\infty}\bigwedge_{k<j}P(\pi[k-1],\pi[k])\wedge||\Phi||(\pi[j]).
\end{equation}

Dually, we have the following GPoCTL state formulae as presented in Eq.(\ref{eq:necessity 1}-\ref{eq:necessity 5}):
\begin{equation}\label{eq:necessity 21}
Ne(\bigcirc \Phi)=\neg Po(\bigcirc\neg\Phi);
\end{equation}
\vspace{-0.6cm}
\begin{equation}\label{eq:necessity 22}
Ne(\Phi \sqcup\Psi)=\neg Po(\neg \Psi\sqcup (\neg\Phi\wedge \neg\Psi))\wedge \neg Po(\square\neg \Phi);
\end{equation}
\vspace{-0.6cm}
\begin{equation}\label{eq:necessity 23}
Ne(\Phi \sqcup^{\leq n}\Psi)=\neg Po(\neg \Psi\sqcup^{\leq n} (\neg\Phi\wedge \neg\Psi))\wedge \neg Po(\square^{\leq n}\neg \Phi);
\end{equation}
\vspace{-0.6cm}
\begin{equation}\label{eq:necessity 24}
Ne(\square\Phi)=\neg Po(\lozenge\neg \Phi);
\end{equation}
\vspace{-0.6cm}
\begin{equation}\label{eq:necessity 25}
Ne(\lozenge\Phi)=\neg Po(\square\neg \Phi).
\end{equation}

\begin{remark}\label{re:comparison of PoCTL and GPoCTL}
By the semantics of GPoCTL, even if we use normal possibilistic Kripke structures as done in \cite{li12}, the semantics of GPoCTL is still not the same as that of PoCTL. The semantics of GPoCTL contains more possibility information. We shall give explicit explanation using some examples in the following section. \end{remark}

\section{GPoCTL model checking}

Similar to multi-valued CTL model-checking problems \cite{chechik012}, the GPoCTL model-checking problem can be stated as follows:

For a given finite generalized possibilistic
Kripke structure $M$, a state $s$ in $M$, and a PoCTL state formula $\Phi$, compute the value $||\Phi||(s)$.

 We write $M\models \Phi$ for this PoCTL model-checking problem.


$||\Phi||(s)$ can be calculated inductively on the length of $\Phi$, $|\Phi|$, i.e., $|\Phi|$ denotes the number of subformulae of $\Phi$, which is defined as follows:

$|\Phi|=1$ if $\Phi\in AP\cup \{true\}$.

$|\Phi\wedge\Psi|=|\Phi|+|\Psi|+1$.

$|\neg\Phi|=|\Phi|+1$.

$|Po(\bigcirc\Phi)|=|Po(\square\Phi)|=|\Phi|+1$.

$|Po(\Phi\sqcup\Psi)|=|Po(\Phi\sqcup^{\leq n}\Psi)|=|\Phi|+|\Psi|+1$.

If $\Phi=a\in AP, \neg\Phi, \Phi_1\wedge \Phi_2$, then we can compute $||\Phi||$ inductively using Eq.(\ref{eq:a}), Eq.(\ref{eq:negation}) and Eq.(\ref{eq:wedge}).
For the formula $\Phi=Po(\varphi)$, where $\varphi$ is a path formula. Since $||\Phi||(s)=Po(s\models \varphi)$, the key point is to calculate $Po(s\models \varphi)$ for any state $s$.

There are four ways  to construct path formula $\varphi$, i.e., $\varphi=\bigcirc\Psi$, $\varphi=\Phi\sqcup^{\leq n}\Psi$, $\varphi=\Phi\sqcup\Psi$  or $\varphi=\square\Psi$ for some state formulae $\Phi$ and $\Psi$ and $n\in \mathbb{N}.$

For $\varphi=\bigcirc\Psi$, the next-step operator, the calculation is as follows,
\begin{eqnarray*}
||Po(\bigcirc\Psi)||(s)&=&Po(s\models\bigcirc\Psi)\\
&=&\bigvee_{\pi\in Paths(s)}Po^{M_s}(\pi)\wedge ||\bigcirc\Psi||(\pi)\\
&=&\bigvee_{\pi=ss_1s_2\cdots\in Paths(s)}P(s,s_1)\wedge P(s_1,s_2)\wedge\cdots \wedge P(s,s_1)\wedge ||\Psi||(s_1)\\
&=&\bigvee_{s_1\in S}P(s,s_1)\wedge||\Psi||(s_1)\wedge (\bigvee_{s_2,s_3,\cdots,\in S}P(s_1,s_2)\wedge P(s_2,s_3)\cdots)\\
&=&\bigvee_{s_1\in S}P(s,s_1)\wedge ||\Psi||(s_1)\wedge r_P(s_1)
\end{eqnarray*}
\noindent where $P$ is the transition matrix of $M$. We will give a matrix representation of the next-step operator. For this purpose, let us first fix some notations. For a state formula $\Phi$, write $D_{\Phi}$ for the $|S|\times |S|$ matrix such that $D_{\Phi}(s,t)=||\Phi||(s)$ if $t=s$ and $0$ otherwise, $D_{\Phi}$ is a diagonal fuzzy matrix with dimension $|S|$ such that the entry $D_{\Phi}(s,s)$ is $||\Phi||(s)$ for any $s\in S$, i.e., $D_{\Phi}=diag(||\Phi||(s))_{s\in S}$. For a function $f:S\longrightarrow [0,1]$, we also use $f$ to represent the column vector corresponding to the function $f$, i.e., $f=(f(s))_{s\in S}$.
In the matrix-vector notation we thus
have that the (column) vector $(Po(s\models\bigcirc\Psi))_{s\in S}$ can be computed by multiplying $P$ with the
vector $D_{\Psi}\circ r_P$, i.e.,  we have
\begin{eqnarray}\label{eq:next operator}
Po(\bigcirc\Psi)=(Po(s\models\bigcirc\Psi))_{s\in S}=P\circ D_{\Psi}\circ r_P.
\end{eqnarray}
It follows that, checking the next-step operator thus reduces to two multiplications of fuzzy matrixes.

To calculate the possibility $Po(s\models\varphi)$ for restricted until formula $\varphi=\Phi\sqcup^{\leq n}\Psi$, we have

\begin{eqnarray*}
||Po(\Phi\sqcup^{\leq n}\Psi)||(s)&=&\bigvee_{\pi=ss_1s_2\cdots\in Paths(s)}Po^{M_s}(\pi)\wedge||\Phi\sqcup^{\leq n}\Psi)||(\pi)\\
&=&\bigvee_{\pi=ss_1s_2\cdots\in Paths(s)}P(s,s_1)\wedge P(s_1,s_2)\cdots\wedge (||\Psi||(s)
\vee\bigvee_{0<j\leq n}(||\Phi||(s)\\
&&\wedge\bigwedge_{k<j}P(s_{k-1},s_k)\wedge||\Phi||(s_k))\wedge P(s_{j-1},s_j)\wedge||\Psi||(s_j))\\
&=&(||\Psi||(s)\wedge r(s))\vee(\bigvee_{0<j\leq n}||\Phi||(s)\wedge\bigwedge_{k<j}P(s_{k-1},s_k)\wedge ||\Phi||(s_k)\\
&&\wedge P(s_{j-1},s_j)\wedge||\Psi||(s_j)\wedge r_P(s_j))\\
&=&(\bigvee_{i=0}^n(D_{\Phi}\circ P)^i\circ D_{\Psi}\circ r_P)(s)
\end{eqnarray*}
In the matrix-notation we have a compact expression as follows,
\begin{equation}\label{eq:expression of restricted until}
||Po(\Phi\sqcup^{\leq n}\Psi)||=(||Po(\Phi\sqcup^{\leq n}\Psi)||(s))_{s\in S}=\bigvee_{i=0}^n(D_{\Phi}\circ P)^i\circ D_{\Psi}\circ r_P.
\end{equation}

If we let $N=|S|$, we know that $\bigvee_{i=0}^n(D_{\Phi}\circ P)^i=(D_{\Phi}\circ P)^{\ast}$, the reflexive and transitive closure of the fuzzy matrix $D_{\Phi}\circ P$, for any $n\geq N$. In this case, we have
\begin{equation}\label{eq:expression of restricted until}
||Po(\Phi\sqcup^{\leq n}\Psi)||=(D_{\Phi}\circ P)^{\ast}\circ D_{\Psi}\circ r_P.
\end{equation}

By the definition of $\Phi\sqcup\Psi$, we can see that $Po(s\models\Phi\sqcup\Psi)=\lim_{n\rw \infty}||Po(\Phi\sqcup^{\leq n}\Psi)||(s)$ for any state $s$. It follows that
\begin{equation}\label{eq:expression of until}
||Po(\Phi\sqcup\Psi)||=(||Po(\Phi\sqcup\Psi)||(s))_{s\in S}=(D_{\Phi}\circ P)^{\ast}\circ D_{\Psi}\circ r_P,
\end{equation}
which can be computed effectively.

To calculate the possibility $Po(s\models\varphi)$ for always operator $\varphi=\square\Phi$, note that

$||\square \Phi||(\pi)=\bigwedge_{i=0}^{\infty}\bigwedge_{j=0}^{i-1} P(\pi([j]),\pi([j+1]))\wedge ||\Phi||(\pi([i]))$,

\noindent then we have, for any state $s$,
\begin{eqnarray*}
Po(s\models \square\Phi)&=&\bigvee_{\pi\in Paths(M)} Po^{M_s}(\pi)\wedge ||\square\Phi||(\pi)\\
&=&\bigvee_{\pi\in Paths(s)}\bigwedge_{j=0}^{\infty}P(\pi([j]),\pi([j+1]))\wedge \bigwedge_{j=0}^{\infty}||\Phi||(\pi([j]))
\end{eqnarray*}
Unlike the next formula and until formula, it is not easy to give a matrix representation of $Po(\square\Phi)$. To give an effective method to compute $Po(\square\Phi)$, we use the fixpoint techniques.

First, let us give an observation.

\begin{proposition}\label{pro:image set}
For any GPoCTL state formula $\Phi$ and a finite GPKS $M$, the image set of $||\Phi||$, denoted by $Im(\Phi)$, is a finite subset of the unit interval [0,1].
\end{proposition}

\begin{proof}
Write $U$ the set of the union of the image set of atomic proposition $a$ and its negation $\neg a$ for $a\in AP$, i.e., $U=\cup\{Im(||a||)\cup Im(||\neg a||) | a\in AP\}$. Since $M$ is a finite GPKS,  $U$ is a finite subset of the unit interval [0,1]. Since the minimum operation and the maximum operation on $U$ do not generate any new elements except the set $U$, the image set of any state formula $\Phi$ is contained in the set $U$. It follows that the image set of $||\Phi||$ is also finite.
\end{proof}

\begin{proposition}\label{pro:image set}
For any GPoCTL state formula $\Phi$ and a finite GPKS $M$, the function defined by $f(Z)=||\Phi||\wedge||Po(\bigcirc Z)||$, where $||Po(\bigcirc Z)||=P\circ D_Z\circ r_P$, which is from the set of possibility distributions over the state set $S$ into itself, has a greatest fixpoint, and the greatest fixpoint of $f$ is just $||Po(\square\Phi)||$.
\end{proposition}

\begin{proof}
Let $Z_0=(1,1,\cdots,1)^T$ be the greatest vector with entries 1. Inductively, we can define $Z_{i+1}=f(Z_i)$. Since $f$ is monotong, i.e., if $Z^{'}\leq Z^{"}$, then $f(Z^{'})\leq f(Z^{"})$, where $Z^{'}\leq Z^{''}$ means that $Z^{'}(s)\leq Z^{''}(s)$ for any state $s$. Then we have the chain $Z_0\geq Z_1\geq Z_2\geq\cdots \geq Z_i\geq Z_{i+1}\geq \cdots$.

Since $Im(||\Phi||)$ is finite, and the operations involved in the function $f$ do not generate any new elements except $U$, it follows that $Im(Z_i)\subseteq U$, which means that $Z_i$ is a function from the state set $S$ into the finite set $U$. Since the set of all the functions from $S$ into $U$ is a finite set, it follows that there exists $k$ such that $Z_{k+1}=Z_k$, i.e., $f(Z_k)=Z_k$. We show that $Z_k$ is the greatest fixpoint of $f$. It is almost obvious that, if $Z$ is a fixpoint of $f$, then $Z\leq Z_0$. Since $f$ is monotone, it follows that $Z=f(Z)\leq Z_1$. Inductively, we have $Z\leq Z_k$. Hence, $Z_k$ is the greatest fixpoint of $f$.

Let $A=||Po(\square\Phi)||$. Then $A$ is defined as,
$A(s)=\bigvee_{\pi\in Paths(s)}\bigwedge_{j=0}^{\infty} P(\pi([j],\pi([j+1]))\wedge \bigwedge_{j=0}^{\infty}||\Phi||(\pi([j]))$, for any state $s$.

First, let us show that $A$ is a fixpoint of $f$.
For any state $s$, we have,
\begin{eqnarray*}
f(A)(s)&=&||\Phi||(s)\wedge ||Po(\bigcirc A)||(s)\\
&=&||\Phi||(s)\wedge \bigvee_{s_1\in S}P(s,s_1)\wedge A(s_1)\\
&=&||\Phi||(s)\wedge \bigvee_{s_1\in S}P(s,s_1)\wedge \bigvee_{\pi\in Paths(s_1)}\bigwedge_{j=1}^{\infty} P(\pi([j],\pi([j+1]))\wedge \bigwedge_{j=1}^{\infty}||\Phi||(\pi([j]))\\
&=&\bigvee_{\pi\in Paths(s_1)}\bigvee_{s_1\in S}P(s,s_1)\wedge\bigwedge_{j=1}^{\infty} P(\pi([j],\pi([j+1]))\wedge \bigwedge_{j=0}^{\infty}||\Phi||(\pi([j]))\\
&=&\bigvee_{\pi\in Paths(s)}\bigwedge_{j=0}^{\infty} P(\pi([j],\pi([j+1]))\wedge \bigwedge_{j=0}^{\infty}||\Phi||(\pi([j]))\\
&=&A(s)
\end{eqnarray*}
Hence, $A$ is a fixpoint of $f$.

Second, we want to show that $A$ is the greatest fixpoint of $f$. If $Z$ is a fixpoint of $f$, i.e., $Z=||\Phi||\wedge ||Po(\bigcirc Z)||=||\Phi||\wedge P\circ D_Z\circ r_P$, then we have,
\begin{eqnarray*}
Z(s)&=&||\Phi||(s)\wedge \bigvee_{s_1\in S}P(s,s_1)\wedge Z(s_1)\wedge r_P(s_1)\\
&\leq& ||\Phi||(s)\wedge \bigvee_{s_1,s_2\in S}P(s,s_1)\wedge ||\Phi||(s_1)\wedge P(s_1,s_2)\wedge Z(s_2)\\
&\leq&  ||\Phi||(s)\wedge \bigvee_{s_1,s_2,\cdots \in S}P(s,s_1)\wedge ||\Phi||(s_1)\wedge P(s_1,s_2)\wedge ||\Phi||(s_2)\wedge P(s_2,s_3)\wedge \cdots\\
&\leq&  \bigvee_{\pi\in Paths(s)}\bigwedge P(\pi([j],\pi([j+1]))\wedge \bigwedge_{j=0}^{\infty}||\Phi||(\pi([j]))\\
&=& A(s)
\end{eqnarray*}
That is to say, $Z\leq A$. Hence, $A=||Po(\square\Phi)||$ is the greatest fixpoint of $f$.
\end{proof}

What is its time complexity of the fixpoint computation of $f(Z)=||\Phi||\wedge ||Po(\bigcirc Z)||$? Let us give some analysis as follows:
The $n$th iteration of the fixpoint computation of $f(Z)(s)=||\Phi||(s)\wedge ||Po(\bigcirc Z)||(s)$ computes the least upper bound of the values of all paths of length $n$ starting from $s$ satisfying $\Phi$. Since the state space $S$ is finite, for any path $\pi$ of length greater than $|S|+1$, there exists a path $\pi^{\prime}$ of length at most $|S|+1$, whose value is above the value $\pi$. Thus, the fixpoint computation converges after at most $|S|+2$ iterations. Note each iteration of fixpoint computation of $f$ involves only the operations of matrix product and the maximum and minimum operations of real numbers, each iteration takes at most $O(|S|^2)$.  Thus, each fixpoint requires $O(|S|^3)$.

This completes the computation of the state formula $Po(\varphi)$.



In the calculation of $(||\Phi||(s))_{s\in S}$ for a state formula $\Phi$, we only need to perform (fuzzy) matrix multiplication at most $|S|+3$ times or perform iteration of fixpoint computation of $f$ at most $|S|+2$ times. It follows that the time complexity of GPoCTL model checking of a finite generalized  possibilistic Kripke structure $M$ and a GPoCTL formula $\Phi$ can be presented as follows.

\begin{theorem}\label{th:time of PoCTL} ({\rm{Time Complexity of GPoCTL Model Checking}})
For a finite possibilistic Kripke structure $M$ and a GPoCTL formula $\Phi$, the GPoCTL model-checking problem $M\models \Phi$ can be determined in time ${\cal O}(size(M)\cdot poly(|S|)\cdot |\Phi|)$, where $|\Phi|$ denotes the number of subformulae of $\Phi$, $poly(N)$ denotes the polynomial function of $N$.
\end{theorem}

The corresponding algorithm can be presented here.

\vspace{-0.1cm}

\line(2,0){335}

\vspace{-0.1cm}

{\bf Algorithm 1:}  Computing the greatest fixpoint

\vspace{-0.1cm}

\line(2,0){335}

\vspace{-0.1cm}

Input:\ \ A function $f$ from the set of possibility distributions over the state set $S$ into itself.

Output:\ \ The greatest fixpoint of $f$.

\vspace{-0.1cm}

\line(2,0){335}

\vspace{-0.1cm}

Procedure \ \  Fixpoint$(x,f)$

\ \ \ $x^{\prime}\longleftarrow f(x)$

\ \ \ while $x\not= x^{\prime}$ do

\ \ \ $x\longleftarrow x^{\prime}$

\ \ \  $x^{\prime}\longleftarrow f(x)$

\ \ \ end while

 return $x$

End Procedure

\vspace{-0.1cm}

\line(2,0){335}

\vspace{-0.1cm}

\line(2,0){335}

\vspace{-0.1cm}

{\bf Algorithm 2:}  GPoCTL Model Checking

\vspace{-0.1cm}

\line(2,0){335}

\vspace{-0.1cm}

Input:\ \ A GPKS $M$ and a GPoCTL formula $\Phi$.

Output:\ \ The possibility $s\models \Phi$, i.e., $||\Phi||(s)$, for every state $s$ in $M$.

\vspace{-0.1cm}

\line(2,0){335}

\vspace{-0.1cm}

Procedure \ \  GPoCTLCheck($\Phi$)

Case \ \  $\Phi$

$true$ \ \   return \ \ $(1)_{s\in S}$

 $a$ \ \  return \ \  $(L(s,a))_{s\in S}$

 $\neg \Phi$ \ \     return \ \  $(1-||\Phi||(s))_{s\in S}$

 $\Phi_1\wedge \Phi_2$ \ \  return \ \  $(||\Phi_1||(s)\wedge ||\Phi_2||(s))_{s\in S}$

 $Po(\bigcirc\Phi)$ \ \ return \ \  $P\circ D_{\Phi}\circ r_P$

 $Po(\Phi_1\sqcup^{\leq n}\Phi_2)$ \ \ return \ \  $\bigvee_{i=0}^n(D_{\Phi_1}\circ P)^i\circ D_{\Phi_2}\circ r_P$

 $Po(\Phi_1\sqcup \Phi_2)$ \ \ return \ \  $(D_{\Phi_1}\circ P)^{\ast}\circ D_{\Phi_2}\circ r_P$

 $Po(\lozenge \Phi)$ \ \  return \ \  $P^{\ast}\circ D_{\Phi}\circ r_P$

 $Po(\square \Phi)$ \ \  return \ \  Fixpoint$((1)_{s\in S},f_{\Phi})$

End Case

End Procedure

\vspace{-0.1cm}

\line(2,0){335}

\vspace{-0.1cm}

Here, $P=(P(s,t))_{s,t\in S}$, $D_{\Phi}=diag(||\Phi(s)||)_{s\in S}$, $r_P=P^{+}\circ D$, $P^+=P\vee P^2\vee\cdots\vee P^N$, $D=(P^+(s,s))_{s\in S}$,  $P^{\ast}=P^0\vee P^+$, where $N=|S|$, and $P^0$ denotes the $N\times N$ identity matrix, $f_{\Phi}(Z)=||\Phi||\wedge P\circ D_Z\circ r_P$. For a vector $r=(r(i))_{i\in I}$, $\neg r=(1-r(i))_{i\in I}$.

\vspace{-0.1cm}

\line(2,0){335}

\vspace{-0.1cm}

We give an example to show the methods of this section.

\begin{example}
{\rm We give some calculations using Example \ref{ex:running example}. For the path formula $\varphi=\bigcirc a$, and for a path $\pi\in Paths(s_0)$, we can simply compute $||\bigcirc a||(\pi)$ as follows:

In Fig.1,
\begin{eqnarray*}
 ||\bigcirc a||(\pi)=\left\{
\begin{array}{ccc}
0.6,& $if$\ \pi\in Cyl(s_0s_1),\\
0.4,& $if$\ \pi\in Cyl(s_0s_3),\\
0,& $otherwise$.\\
\end{array}
\right.
 \end{eqnarray*}
 In Fig.2, \begin{eqnarray*}
 ||\bigcirc a||(\pi)=\left\{
\begin{array}{ccc}
0.6,& $if$\ \pi\in Cyl(s_0s_1),\\
0.4,& $if$\ \pi\in Cyl(s_0s_3),\\
0,& $otherwise$.\\
\end{array}
\right.
 \end{eqnarray*}
 In Fig.3,
 \begin{eqnarray*}
 ||\bigcirc a||(\pi)=\left\{
\begin{array}{ccc}
0.8,& $if$\ \pi\in Cyl(s_0s_1),\\
1,& $if$\ \pi\in Cyl(s_0s_3),\\
0,& $otherwise$.\\
\end{array}
\right.
 \end{eqnarray*}
 In Fig.4, \begin{eqnarray*}
 ||\bigcirc a||(\pi)=\left\{
\begin{array}{ccc}
0.8,& $if$\ \pi\in Cyl(s_0s_1),\\
0.9,& $if$\ \pi\in Cyl(s_0s_3),\\
0,& $otherwise$.\\
\end{array}
\right.
 \end{eqnarray*}
 We can see that even in a PKS as in Fig.3, the path formula $\bigcirc a$ in GPoCTL is not crisp. As we know, all formulae in PoCTL, including state and path formulae, are crisp, see \cite{li13}. The semantics of GPoCTL, compared with that of PoCTL, contains more possibility information.
 Furthermore, using Algorithm 2, we can give the semantics of GPoCTL formulae $Po(\bigcirc (a\wedge b))$ and $Po(b\sqcup c)$ in the GPKS as shown in Fig.1 as follows, where $X^T$ denotes the transposed fuzzy matrix of $X$.

 $Po(\bigcirc (a\wedge b))=(Po(s\models \bigcirc (a\wedge b)))_{s\in S}=P\circ D_{a\wedge b}\circ r_P=(0.5,0.4,0,0.5)^T$,

 $Po(b\sqcup c)=(Po(s\models b\sqcup c))_{s\in S}=(D_b\circ P)^{\ast}\circ D_c\circ r_P=(0.6,0.5,0.7,0.6)^T$,

 \noindent where $P=\left(\begin{array}{cccc}
0&0.8&0&0.9\\
0&0&0.2&0.5\\
0&0&0.9&0\\
0&0.7&0.6&0.4
\end{array}
\right)$,
$D_{a\wedge b}=diag(0.8,0.6,0,0.4)$, $D_b=diag(0.8,1,0,0.5)$, $D_c=diag(0,0,0.7,1)$ and $r_P=(0.6,0.5,0.9,0.6)^T$.}

\end{example}

\section{Semantics interpretation of GPoCTL in possibilistic Kripke structures and restricted GPoCTL}

Another view of quantitative GPoCTL model checking can be presented as follows: For a given interval $J\subseteq [0,1]$, and for a state formula $\Phi$ in GPoCTL, determine whether $||\Phi||(s)\in J$ for any state $s\in S$. Corresponding to this model checking, a related crisp formula $\Phi_J$ is defined using the semantics of $\Phi$ under a GPKS $M$ as, \begin{equation}
s\models \Phi_J \ {\rm iff} \ ||\Phi||(s)\in J.
 \end{equation}
 In fact, the formula $\Phi_J$ can be decided by the model-checking algorithm in the above section.

Concretely, for an atomic formula $a$ in $AP$,  states formulae $\Phi,\Psi$,  and a path formula $\varphi$, we have
\begin{eqnarray*}
s\models a_J  & {\rm iff} \ L(s,a)\in J;\\
s\models(\neg\Phi)_J & {\rm iff}\ 1-||\Phi||(s)\in J; \\
s\models(\Phi\wedge\Psi)_J  & {\rm iff} \ ||\Phi||(s)\wedge ||\Psi||(s)\in J;\\
s\models (Po(\varphi))_J  & {\rm iff} \ Po(s\models \varphi)\in J, \ {\rm we \ write}\ Po_J(\varphi) \ {\rm as} \ (Po(\varphi)_J \ {\rm in \ the \ sequel}.
\end{eqnarray*}

The formula $\Phi_J$ is very similar to that used in PoCTL. We shall study the relationship between GPoCTL and PoCTL. For this purpose, we shall restrict the GPKS to PKS when we talk about the semantics of GPoCTL, since we only consider the semantics of PoCTL in the frame of PKS. In this case, we shall see the much more simple form of $\Phi_J$.

In this section, all GPKS considered will be PKS. We have the following basic results.

\begin{definition}\label{de:equivalence of state formulae}
For two state formulae $\Phi$ and $\Psi$ in GPoCTL, and any interval $J,K\subseteq [0,1]$, $\Phi_J\equiv \Psi_K$ iff ``$s\models \Phi_J$ iff $s\models \Psi_K$'' holds for any PKS $M$.
\end{definition}

\begin{lemma}\label{le:equivalence of atomic formulae}
For any $a\in AP$, (1) $a_{[0,1]}\equiv true$, (2) $a_J\equiv a_{[0,0]}=\neg a$ for any interval $0\in J\subseteq [0,1)$, (3) $a_J\equiv a_{[1,1]}=a$ for any interval $1\in J\subseteq (0,1]$.
\end{lemma}

\begin{proof}
For any PKS $M$ and any state $s$ in $M$, we have the following observation.

(1) $s\models a_{[0,1]}$ iff $||a||(s)\in [0,1]$. Since $||a||(s)\in [0,1]$ always holds, it follows that $s\models a_{[0,1]}$. Note that $s\models true$ holds for any state $s$, we then have $a_{[0,1]}\equiv true$.

(2) $s\models a_{[0,0]}$ iff $||a||(s)=0$ iff $a\not\in L(s)$ iff $s\not\models a$ iff $s\models \neg a$. Note that for a PKS $M$, the labeling function $L$ is crisp, i.e., $||a||(s)=L(s,a)=0$ or $1$, it follows that, for any interval $J\subseteq [0,1)$ such that $0\in J$, $||a||(s)\in J$ iff $||a||(s)=0$, i.e., $s\models a_J$ iff $||a||(s)=0$. Hence, $a_J\equiv a_{[0,0]}=\neg a$ for any interval $0\in J\subseteq [0,1)$.

(3) $s\models a_{[1,1]}$ iff $||a||(s)=1$ iff $a\in L(s)$ iff $s\models a$. Note that for a PKS $M$, the labeling function $L$ is crisp, i.e., $||a||(s)=L(s,a)=0$ or $1$, it follows that, for any interval $J\subseteq (0,1]$ such that $1\in J$, $||a||(s)\in J$ iff $||a||(s)=1$, i.e., $s\models a_J$ iff $||a||(s)=1$. Hence, $a_J\equiv a_{[1,1]}=a$ for any interval $1\in J\subseteq (0,1]$.
\end{proof}

By the above lemma, we can write $a$ as $a_{[1,1]}$ and $\neg a$ as $a_{[0,0]}$. Then it holds that $s\models a$ iff $a\in L(a)$ and $s\models \neg a$ iff $a\not\in L(s)$. From atomic formulae $a$ in $AP$, we can infer any state formulae of PoCTL from state formulae of GPoCTL, as presented in the following two theorems.

\begin{theorem}\label{th:PoCTL is a proper part of GPoCTL}
For any state formula $\Phi$ in GPoCTL, and any interval $J\subseteq [0,1]$ with rational bounds, $\Phi_J$ is a state formula of PoCTL, i.e., there is an equivalent state formula $\Psi$ in PoCTL such that $\Phi_J\equiv \Psi$.
\end{theorem}

\begin{proof} The proof is proceeded inductively on the length of formula $\Phi$, $|\Phi|$.
For any PKS $M$ and any state $s$ in $M$, we have the following discussion.

If $|\Phi|=1$, then $\Phi=a\in AP$ or $\Phi=true$, by Lemma \ref{le:equivalence of atomic formulae}, $\Phi_J$ is a PoCTL state formula.

Assume that $\Phi_J$ is a PoCTL state formula for any GPoCTL state formula $\Phi$ with length $|\Phi|\leq n$. For a GPoCTL formula $\Phi$ with length $n+1$, we want to show that $\Phi_J$ is a PoCTL state formula for any interval $J$. There are four forms of the interval $J$, that is, $J=[u,v], (u,v], [u,v)$ or $(u,v)$ for $u,v\in [0,1]$. We give the proof for the case of the closed interval $J=[u,v]$, other cases are completely the same and thus omitted. In the following, $J$ is always the closed interval $[u,v]$.

There are six cases to be considered.

Case 1: $\Phi=\pp\wedge \ppp$ for two GPoCTL state formulae $\pp$ and $\ppp$.

Write $\Phi_{\geq u}=\Phi_{[u, 1]}$ and $\Phi_{\leq v}=\Phi_{[0,v]}$. Since $\Phi_{[u,v]}=\Phi_{\geq u}\wedge \Phi_{\leq v}$, it suffices to calculate $\Phi_{\geq u}$ and $\Phi_{\leq v}$.

Note that $s\models \Phi_{\geq u}$ iff $||\pp||(s)\wedge ||\ppp||(s)\geq u$ iff $||\pp||(s)\geq u$ and $||\ppp||(s)\geq u$ iff $s\models \pp_{\geq u}$ and $s\models \ppp_{\geq u}$ iff $s\models \pp_{\geq u} \wedge \ppp_{\geq u}$.

Therefore, $\Phi_{\geq u}\equiv\pp_{\geq u}\wedge \ppp_{\geq u}$.

Note that $s\models \Phi_{\leq v}$ iff $||\pp||(s)\wedge ||\ppp||(s)\leq v$ iff $||\pp||(s)\leq v$ or $||\ppp||(s)\leq v$ iff $s\models \pp_{\leq v}$ or $s\models \ppp_{\leq v}$ iff $s\models \pp_{\leq v} \vee \ppp_{\leq v}$.

Therefore, $\Phi_{\leq v}\equiv\pp_{\leq v}\vee \ppp_{\leq v}$.

Hence, $\Phi_J=\Phi_{\geq u}\wedge \Phi_{\leq v}\equiv (\pp_{\geq u}\wedge \ppp_{\geq u})\wedge (\pp_{\leq v}\vee \ppp_{\leq v})$. By the induction, we know that $\Phi_J$ is a PoCTL state formula.

Case 2: $\Phi=\neg \pp$ for a GPoCTL formula $\pp$.

Note that, $s\models \Phi_J$ iff $u\leq ||\Phi||(s)\leq v$, iff $u\leq ||\neg\pp||(s)\leq v$, iff $u\leq 1-||\pp||(s)\leq v$, iff $1-v\leq ||\pp||(s)\leq 1-u$, iff $s\models \pp_{[1-v,1-u]}$.

Therefore, $\Phi_J\equiv \pp_{DJ}$. By the induction, we have $\Phi_J$ is a PoCTL state formula.

Case 3: $\Phi=Po(\bigcirc\pp)$.

Note that $s\models Po_{\geq u}(\bigcirc \pp)$ iff $\bigvee_{s_1\in S}P(s,s_1)\wedge ||\pp||(s_1)\geq u$, iff there exists a state $s_1$ such that $P(s,s_1)\geq u$ and $\pp(s_1)\geq u$, iff there exists a state $s_1$ such that $P(s,s_1)\geq u$ and $s_1\models \pp_{\geq u}$, iff $Po^{M_s}(\{\pi\in Paths(s) | \pi\models \bigcirc \pp_{\geq u}\})\geq PO^{M_s}(Cyl(ss_1))=P(s,s_1)\geq u$, iff $s\models Po_{\geq u}(\bigcirc \pp_{\geq u})$.

Therefore, $Po_{\geq u}(\bigcirc\pp)\equiv Po_{\geq u}(\bigcirc \pp_{\geq u})$.

Note that $s\models Po_{\leq v}(\bigcirc \pp)$ iff $\bigvee_{s_1\in S}P(s,s_1)\wedge ||\pp||(s_1)\leq v$, iff for any state $s_1$, we have $P(s,s_1)\wedge \pp(s_1)\leq v$, iff for any state $s_1$, $P(s,s_1)\leq v$ or $||\pp||(s_1)\leq v$, iff for any state $s_1$, if $||\pp||(s_1)>v$, then $P(s,s_1)\leq v$, iff  for any state $s_1$, if $s_1\models \pp_{>v}$, then $P(s,s_1)\leq v$, iff $Po^{M_s}(\{\pi\in Paths(s) | \pi\models \bigcirc \pp_{>v}\})=Po^{M_s}(\cup\{Cyl(ss_1) | s_1\models\pp_{>v}\}=\bigvee\{P(s,s_1) | s_1\models\pp_{>v}\}\leq v$, iff $s\models Po_{\leq v}(\bigcirc \pp_{>v})$, where $\pp_{>v}=\pp_{(v,1]}$.

Therefore, $Po_{\leq v}(\bigcirc \pp)\equiv Po_{\leq v}(\bigcirc \pp_{>v})$.

Hence, $Po_J(\bigcirc \pp)=Po_{\geq u}(\bigcirc \pp)\wedge Po_{\leq v}(\bigcirc \pp)\equiv Po_{\geq u}(\bigcirc\pp_{\geq u})\wedge Po_{\leq v}(\bigcirc\pp_{>v})$. By the induction, we know that $Po_J(\bigcirc \pp)$ is a PoCTL state formula.

Case 4: $\Phi=Po(\pp\sqcup \ppp)$.

Note that $s\models \Phi_{\geq u}$ iff there exists a path $\pi=s_0s_1\cdots$, and the integer $j$, such that $\bigwedge_{k\leq j}P(s_{k-1},s_k)\wedge \bigwedge_{k<j}||\pp||(s_k)\wedge ||\ppp||(s_j)\geq u$, iff there exists a path $\pi=s_0s_1\cdots$, and a $j$, such that $\bigwedge_{k\leq j}P(s_{k-1},s_k)\geq u$ and $\wedge \bigwedge_{k<j}||\pp||(s_k)\wedge ||\ppp||(s_j)\geq u$, iff there exists a path $\pi\in Cyl(s_0\cdots s_j)$ such that $\bigwedge_{k\leq j}P(s_{k-1},s_k)\geq u$ and $\pi\models \pp_{\geq u}\sqcup \ppp_{\geq u}$, iff $\bigvee\{Po^{M_s}(\pi) | \pi\in Paths(s), \pi\models \pp_{\geq u}\sqcup \ppp_{\geq u}\}\geq u$, iff $s\models Po_{\geq u}(\pp_{\geq u}\sqcup \ppp_{\geq u})$.

Therefore, $Po_{\geq u}(\pp\sqcup \ppp)\equiv Po_{\geq u}(\pp_{\geq u}\sqcup \ppp_{\geq u})$.

Note that $s\models \Phi_{\leq v}$ iff, for any path $\pi=s_0s_1\cdots, ||\pp||(s)\leq v$, and for any $j$,  $\bigwedge_{k\leq j}P(s_{k-1},s_k)\wedge \bigwedge_{k<j}||\pp||(s_k)\wedge ||\ppp||(s_j)\leq v$, iff for any path $\pi=s_0s_1\cdots, ||\pp||(s)\leq v$, and for any $j$,  $\bigwedge_{k\leq j}P(s_{k-1},s_k)\leq v$ or $ \bigwedge_{k<j}||\pp||(s_k)\wedge ||\ppp||(s_j)\leq v$, iff for any path $\pi=s_0s_1\cdots, ||\pp||(s)\leq v$, and for any $j$,  if $ \bigwedge_{k<j}||\pp||(s_k)\wedge ||\ppp||(s_j)> v$, then $\bigwedge_{k\leq j}P(s_{k-1},s_k)\leq v$, iff for any path $\pi=s_0s_1\cdots, ||\pp||(s)\leq v$, and for any $j$,  if $\pi\models \pp_{>v}\sqcup \ppp_{>v}$, then $\bigwedge_{k\leq j}P(s_{k-1},s_k)\leq v$, iff $s\models \ppp_{\leq v}$, and $\bigvee\{Po^{M_s}(\pi) | \pi\in Paths(s), \pi\models \pp_{>v}\sqcup \ppp_{>v}\}\leq v$, iff $s\models Po_{\leq v}(\pp_{>v}\sqcup \ppp_{>v})\wedge \ppp_{\leq v}$.

Therefore, $Po_{\leq v}(\pp\sqcup \ppp)\equiv Po_{\leq v}(\pp_{>v}\sqcup \ppp_{>v})\wedge \ppp_{\leq v}$.

Hence, $Po_J(\pp\sqcup \ppp)\equiv (Po_{\geq u}(\pp_{\geq u}\sqcup \ppp_{\geq u}))\wedge (Po_{\leq v}(\pp_{>v}\sqcup \ppp_{>v})\wedge \ppp_{\leq v})$. By the induction, we know that $Po_J(\pp\sqcup \ppp)$ is a GPoCTL state formula.

Case 5: $\Phi=Po(\pp\sqcup^{\leq n} \ppp)$.

Similar to case 4, we have $Po_J(\pp\sqcup^{\leq n} \ppp)\equiv (Po_{\geq u}(\pp_{\geq u}\sqcup^{\leq n} \ppp_{\geq u}))\wedge (Po_{\leq v}(\pp_{>v}\sqcup^{\leq n} \ppp_{>v})\wedge \ppp_{\leq v})$. By the induction, we know that $Po_J(\pp\sqcup^{\leq n} \ppp)$ is a GPoCTL state formula.

Case 6: $\Phi=\square\pp$.

Note that $s\models Po_{\geq u}(\square\pp)$ iff $\bigvee_{\pi\in Paths(s)}\bigwedge_{j=0}^{\infty}P(s_j,s_{j+1})\wedge \bigwedge_{i=0}^{\infty}||\pp||(s_j)\geq u$, iff there exist a path $\pi=s_0s_1\cdots$ with $s_0=s$ such that $\bigwedge_{j=0}^{\infty}P(s_j,s_{j+1})\geq u$ and $\bigwedge_{i=0}^{\infty}||\pp||(s_j)\geq u$, iff there exist a path $\pi=s_0s_1\cdots$ with $s_0=s$ such that $Po^{M_s}(\pi)\geq u$ and $\pi\models \square\pp_{\geq u}$, iff $\bigvee\{Po^{M_s}(\pi) | \pi\in Paths(s), \pi\models \square\pp_{\geq u}\}\geq u$, iff $s\models Po_{\geq u}(\square\pp_{\geq u})$.

Therefore, $Po_{\geq u}(\square\pp)=Po_{\geq u}(\square\pp_{\geq u})$.

Note that $s\models Po_{\leq v}(\square\pp)$ iff $\bigvee_{\pi\in Paths(s)}\bigwedge_{j=0}^{\infty}P(s_j,s_{j+1})\wedge \bigwedge_{i=0}^{\infty}||\pp||(s_j)\leq v$, iff for any path $\pi=s_0s_1\cdots$ with $s_0=s$, $\bigwedge_{j=0}^{\infty}P(s_j,s_{j+1})\wedge \bigwedge_{i=0}^{\infty}||\pp||(s_j)\leq v$, iff for any path $\pi=s_0s_1\cdots$ with $s_0=s$, $\bigwedge_{j=0}^{\infty}P(s_j,s_{j+1})\leq v$ or $\bigwedge_{i=0}^{\infty}||\pp||(s_j)\leq v$, iff for any path $\pi=s_0s_1\cdots$ with $s_0=s$, if $\bigwedge_{i=0}^{\infty}||\pp||(s_j)> v$, then $\bigwedge_{j=0}^{\infty}P(s_j,s_{j+1})\leq v$, iff for any path $\pi=s_0s_1\cdots$ with $s_0=s$, if $\pi\models \square\pp_{>v}$, then $Po^{M_s}(\pi)\leq v$, iff $\bigvee\{ Po^{M_s}(\pi) | \pi\in Paths(s), \pi\models \square\pp_{>v}\}\leq v$, iff $s\models Po_{\leq v}(\square\pp_{>v})$.

Therefore, $Po_{\leq v}(\square\pp)\equiv Po_{\leq v}(\square\pp_{>v})$.

Hence, $Po_J(\square\pp)\equiv Po_{\geq u}(\square\pp_{\geq u})\wedge Po_{\leq v}(\square\pp_{>v})$. By the induction, we have $Po_J(\square\pp)$ is a PoCTL state formula.
\end{proof}

For a GPoCTL state formula $\Phi$ and interval $J\subseteq [0,1]$, when we use $\Phi_J$ as a state formula and give its semantics in PKS, we can get a restricted version of GPoCTL as defined as follows.

\begin{definition}\label{de:RGPoCTL}
(Syntax of RGPoCTL) Restricted generalized possibilistic CTL (RGPoCTL, in short) state formulae over the set $AP$ of atomic propositions are formed according to the following grammar:
\begin{center}
$\Phi ::= true\mid  a \mid\Phi_{1} \wedge \Phi_{2}\mid \neg \Phi\mid Po_J(\varphi)$
\end{center}
where $a\in AP$, $\varphi$ is a RPoCTL path formula and $J$ is an interval of $[0,1]$ with rational bounds.

RPoCTL path formulae are formed according to the following grammar:

\begin{center}
$\varphi::=\bigcirc \Phi \mid \Phi_{1}\sqcup \Phi_{2}\mid \Phi_{1}\sqcup^{\leq n} \Phi_{2}\mid \square\Phi$
\end{center}
where $\Phi$, $ \Phi_{1}$, and $ \Phi_{2}$ are state formulae and $n\in\mathbb{N}$.
\end{definition}

The semantics of RGPoCTL formulae is interpreted in PKS.
Let $a\in AP$ be an atomic proposition, $M=(S,P,I,AP,L)$ be a possibilistic Kripke structure, $s$ be a state, $\Phi$, $\Psi$ be RGPoCTL state formulae, and $\varphi$ be a RGPoCTL path formula. The satisfaction relation $\models$ is defined for state formulae by,
\begin{eqnarray*}
s\models a  & {\rm iff} \ a\in L(s);\\
s\models\neg\Phi & {\rm iff}\ s\not\models\Phi; \\
s\models\Phi\wedge\Psi  & {\rm iff} \ s\models\Phi \ {\rm and}\ s\models\Psi;\\
s\models Po_{J}(\varphi)  & {\rm iff} \ Po(s\models \varphi)\in J.
\end{eqnarray*}

\noindent where $Po(s\models \varphi)=\bigvee\{Po^{M_s}(\pi)\wedge ||\Phi||(\pi) | \pi\in Paths(s)\}$. For path formula $\varphi$, and $\pi\in Paths(M)$, its semantics is a fuzzy set $||\varphi||: Paths(M)\rw [0,1]$, which is defined recursively as follows,

\begin{eqnarray*}
||\bigcirc\Phi||(\pi)= \left\{ \begin{array}{ll} P(\pi([0]),\pi([1])), & {\rm if} \ \pi\models \bigcirc \Phi;\\
0, & {\rm otherwise}.
\end{array}\nonumber
\right.
\end{eqnarray*}
\begin{equation*}
||\Phi\sqcup\Psi||(\pi)=\bigvee_k\{\bigwedge_{j=0}^k P(s_j,s_{j+1})| {\rm for\ any}\ j<k, s_j\models \Phi, {\rm and}\ s_k\models \psi\};
\end{equation*}
\begin{eqnarray*}
||\Phi\sqcup^{\leq n}\Psi||(\pi)=\bigvee_{k\leq n}\{\bigwedge_{j=0}^k P(s_j,s_{j+1})| {\rm for\ any}\ j<k, s_j\models \Phi, {\rm and}\ s_k\models \psi\};
\end{eqnarray*}
\begin{equation*}
||\square\Phi||(\pi)= \left\{ \begin{array}{ll} \bigwedge_{j=0}^{\infty}P(s_j,s_{j+1}), & {\rm if} \ \pi\models \square \Phi;\\
0, & {\rm otherwise}.
\end{array}
\right.
\end{equation*}

\begin{theorem}\label{th:RGPoCTL=PoCTL}
The state formulae of RGPoCTL are the same as those of PoCTL.
\end{theorem}

\begin{proof}
From the definition of state formulae in PoCTL and RGPoCTL, we know that they have the same atomic foumulae $a\in AP$. The left is to show that they have the same state formula $Po_J(\varphi)$ for a path formula $\varphi$ and an interval $J$.

We use the superior $^r$ to represent RGPoCTL formula, and $^p$ to represent the PoCTL formula. It is sufficient to show that $Po^r_J(\varphi)=Po^p_J(\varphi)$ for the same path formula $\varphi$ (but with different semantics). This can be guaranteed by the fact $Po^r(s\models\varphi)=Po^p(s\models\varphi)$, where $\varphi$ has four forms, $\bigcirc\Phi$, $\Phi\sqcup \Psi$, $\Phi\sqcup^{\leq n} \Psi$ and $\square\Phi$. We prove the later in four cases as follows.

Case 1: $\varphi=\bigcirc\Phi$. In this case, $Po^r(s\models\varphi)=\bigvee\{Po^{M_s}(\pi)\wedge ||\bigcirc\Phi||(\pi) | \pi\in Paths(s)\}=\bigvee\{Po^{M_s}(\pi)\wedge P(s,\pi([1])) | \pi\in Paths(s)\}=\bigvee\{Po^{M_s}(\pi) | \pi\in Paths(s), \pi\models \bigcirc\Phi \}=Po^p(s\models\varphi)$.

Case 2: $\varphi=\Phi\sqcup \Psi$. In this case, $Po^r(s\models\varphi)=\bigvee\{Po^{M_s}(\pi)\wedge ||\Phi\sqcup \Psi||(\pi) | \pi\in Paths(s)\}=\bigvee\{Po^{M_s}(\pi) | \pi\in Paths(s), \pi\models \Phi\sqcup \Psi \}=Po^p(s\models\varphi)$.

Case 3: $\varphi=\Phi\sqcup^{\leq n} \Psi$. The proof is similar to that of the case 2.

Case 4: $\varphi=\square\Phi$. In this case, $Po^r(s\models\varphi)=\bigvee\{Po^{M_s}(\pi)\wedge \bigwedge_{j=0}^{\infty} P(\pi([j]),\pi([j+1])) | \pi\in Paths(s), \pi\models \square\Phi \}=\bigvee\{Po^{M_s}(\pi) | \pi\in Paths(s), \pi\models \square\Phi \}=Po^p(s\models\varphi)$.

Since the above fact, for a RGPoCTL path formula or a PoCTL path formula $\varphi$, we write $Po^r(s\models\varphi)$ and $Po^p(s\models\varphi)$ with the same symbol $Po(s\models \varphi)$, which have the same interpretation $Po^{M_s}(\{\pi\in Paths(s) |\pi\models \varphi\})$ for any PKS $M$.

Since $Po^r(s\models\varphi)=Po^p(s\models\varphi)$ for any state $s$ for any PKS $M$, it follows that $Po^r_J(\varphi)=Po^p_J(\varphi)$ for any path formula $\varphi$ and interval $J$. Hence, RGPoCTL and PoCTL have the same state formulae.

\end{proof}

 RGPoCTL and PoCTL have the same state formulae, but with different semantics of path formuae. In this sense, PoCTL can be seen as a qualitative version or a crisp counterpart of GPoCTL, where we interpret GPoCTL formulae in the frame of PKS models.

Moreover, if we further restrict the interval $J\subseteq [0,1]$ with the form $(0,1]$ (write $>0$ in short) and $[1]$ (write $=1$ in short), then we obtain a more narrow qualitative GPoCTL, which is the same as qualitative PoCTL as defined in \cite{li13}, where the system models are PKS models. In this case, CTL is a  proper subclass of qualitative PoCTL (as shown in \cite{li13}), and thus, CTL is a proper subclass of GPoCTL.





\section{An illustrative example}

We consider the thermostat example given in \cite{chechik012,wu07}. A little revision is adopted for our purpose.

There are three models for the thermostat as shown in Fig.5. Fig.5(a) is a very
simple thermostat that can run a heater if the temperature falls below a desired threshold. The system
has one indicator ($Below$), a switch to turn it off and on ($Running$) and a variable indicating whether the heater is running ($Heat$). The system starts in state
$OFF$ and transits into $IDLE1$ when it is turned on, where it awaits the reading of the temperature
indicator. When the temperature is determined, the system transits either into  $IDLE2$ or into $HEAT$. The value of the temperature indicator is unknown in states $OFF$ and $IDLE1$.
We use three-valued GPKS: 1, 0 and 0.5 (Maybe), to model the system,
assigning $Below$ the value 0.5 in states $OFF$ and $IDLE1$ since the temperature is not determined in these two states, as depicted in Fig.5(a). Note that each state in this and the other two systems in Fig.5 contains a self-loop with the value $1$
which we omitted to avoid clutter.

Fig.5(b) shows another aspect of the thermostat system-running the air conditioner.
The behavior of this system is similar to that of the heater, with one difference:
this system handles the failure of the temperature indicator. If the temperature reading
cannot be obtained in states $AC$ or $IDLE2$, the system transits into state $IDLE1$.

Finally, Fig.5(c) gives a combined model, describing the behavior of the thermostat
that can run both the heater and the air conditioner. In this model, we use the same
three-valued GPKS. When the individual descriptions agree that the value
of a variable or transition is 1 (resp., 0), it is mapped into 1 (resp., 0) in the combined model; all other
values are mapped into 0.5.

For simplicity, we use the symbols $r,b,a,ac,h$ to represent the atomic propositions $Running$, $Below$, $Above$, $AC$ and $Heat$.

\begin{figure}[ht]
\begin{center}
\includegraphics[scale=1]{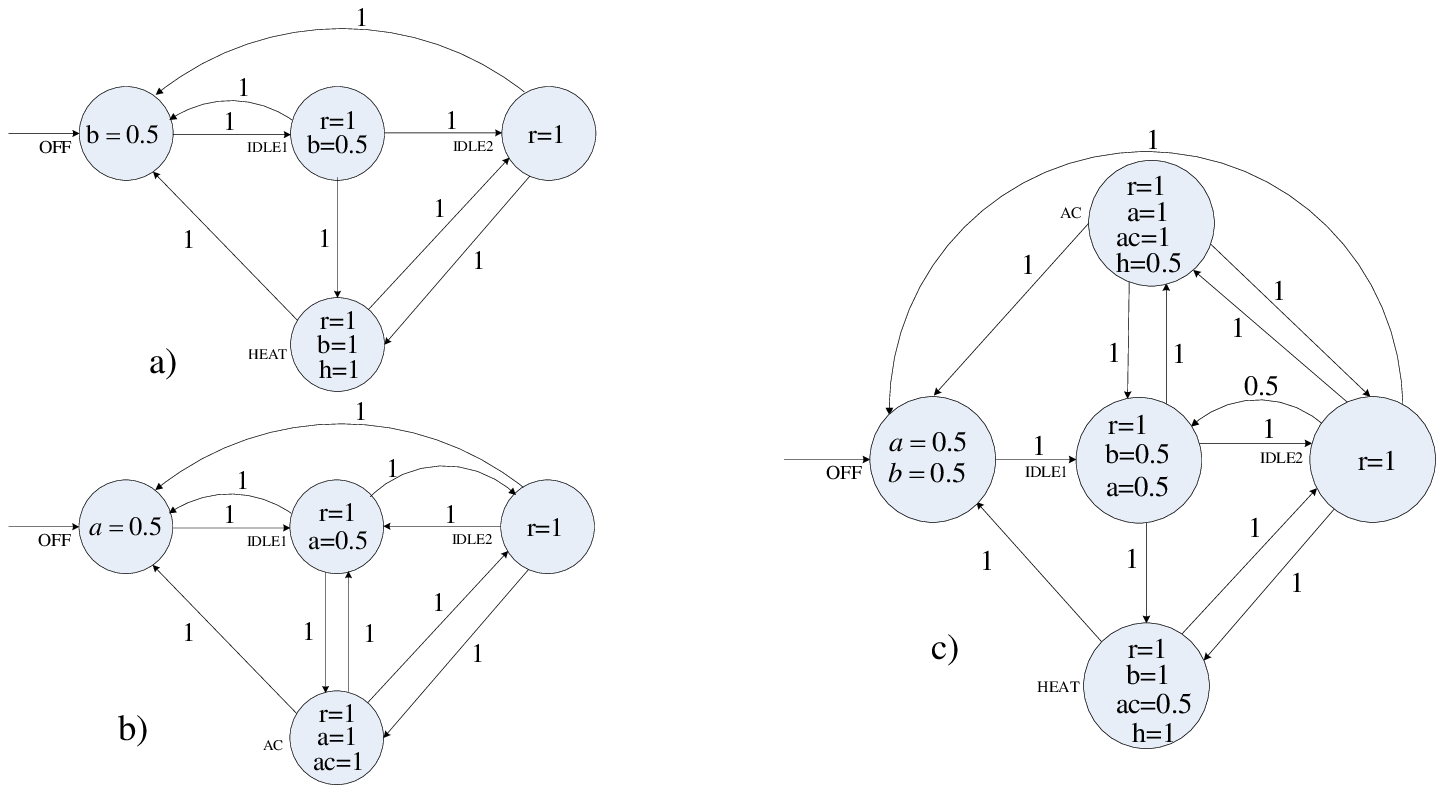}
\center{Fig.5.}Models of the thermostat. (a) Heat only; (b) AC only; (c) combined model.
\vspace{-0.3cm}
\end{center}
\end{figure}

For this thermostat model, we can ask a number of questions as presented in \cite{chechik012}:

Prop. 1. Can the system transit into $IDLE1$ from everywhere?

Prop. 2.  Can the heater be turned on when the temperature falls below a desired threshold?

Prop. 3. Can the system be turned off in every computation?

Prop. 4. Is heat on only if air conditioning is off?

Prop. 5. Can heat be on when the temperature is above a threshold desired?

 The above properties can be re-stated using possibility measures as follows:

Prop. 1p. What is the possibility (resp. necessity) that the system can transit into $IDLE1$ from everywhere?

Prop. 2p.  What is the possibility (resp. necessity) that the heater can be turned on when the temperature falls below a desired threshold?

Prop. 3p. What is the possibility (resp. necessity) that the system can be turned off in every computation?

Prop. 4p. What is the possibility (resp. necessity) that heat is on only if air conditioning is off?

Prop. 5p. What is the possibility (resp. necessity) that heat can be on when the temperature is above a threshold desired?

The above properties can be
described using GPoCTL formulae as presented in Table 1 and Table 2, respectively. The table also lists the values of these properties
in each of the models given in Fig.5. We use ``--'' to indicate that the result cannot be
obtained from this model. For example, the two individual models disagree on the question
of reachability of state $IDLE1$
from every state in the model, whereas the combined
model concludes that it is $0$. We obtain more useful information than those presented in \cite{chechik012,wu07}.


\vspace{-0.5cm}
\begin{center}
\begin{tabular}{|c|c|c|c|c|}
  \hline
  Property & GPoCTL formula & Heat model & AC model & Combined model \\
  \hline
  Prop.1p & $Po(\square Po(\bigcirc IDLE1))$ & (1,1,0,0) & (1,1,1,1) & (1,1,0.5,1,0) \\
  Prop.2p & $Po(\neg Heat \sqcup Below)$ & (1,1,1,1) & -- & (1,1,1,0.5,1) \\
  Prop.3p & $Po(\square Po(\lozenge \neg Runing)))$ & (1,1,1,1) & (1,1,1,1) & (1,1,1,1,1) \\
  Prop.4p & $Po(\square(\neg Ac\rw Heat))$ & -- & -- & (0,0,0,1,1) \\
  Prop.5p & $Po(\square(Above\rw \neg Heat))$ & -- & -- & (1,1,1,0.5,1) \\
 \hline
\end{tabular}
\vspace{-0.4cm}
\center{Table 1.} Results of verifying properties of the thermostat system using possibility measure.

\end{center}

\vspace{-0.5cm}
\begin{center}
\begin{tabular}{|c|c|c|c|c|}
  \hline
  Property & GPoCTL formula & Heat model & AC model & Combined model \\
  \hline
  Prop.1p & $Ne(\square Ne(\bigcirc IDLE1))$ & (0,0,0,0) & (0,0,0,0) & (0,0,0,0,0) \\
  Prop.2p & $Ne(\neg Heat \sqcup Below)$ & (0.5,0.5,0,1) & -- & (0.5,0.5,0,0,1) \\
  Prop.3p & $Ne(\square Ne(\lozenge \neg Runing)))$ & (0,0,0,0) & (0,0,0,0) & (0,0,0,0,0) \\
  Prop.4p & $Ne(\square(\neg Ac\rw Heat))$ & -- & -- & (0,0,0,0,0) \\
  Prop.5p & $Ne(\square(Above\rw \neg Heat))$ & -- & -- & (0.5,0.5,0.5,0.5) \\
 \hline
\end{tabular}
\vspace{-0.4cm}
\center{Table 2.} Results of verifying properties of the thermostat system using necessity measure.

\end{center}

As an illustrative example, let us show how to compute Prop.1p. Let $\Phi=Po(\bigcirc IDLE1)$, then $Po(\square Po(\bigcirc IDLE1))=Po(\square \Phi)$. By Algorithm 2, we have $||\Phi||=P_i\circ D_{IDLE1}^i\circ r_{P_i}$, and $||Po(\square\Phi)||$ is the greatest fixpoint of the operator $f(Z)=||\Phi||\wedge P\circ D_Z\circ r_P$, where $i=a,b,c$ denote GPKSs as shown in Fig.5(a)-(c). By a simple calculation, we have $Po(\square Po(\bigcirc IDLE1))=(1,1,0,0)$ for GPKS in Fig.5(a), $Po(\square Po(\bigcirc IDLE1))=(1,1,1,1)$ for GPKS in Fig.5(b), and $Po(\square Po(\bigcirc IDLE1))=(1,1,0.5,1,0)$ for GPKS in Fig.5(c).
It means that the system shown in Fig.5.(a) can transit into $IDLE1$ from the state $OFF$ (with possibility 1) and $IDLE1$ (with possibility 1) and could not transit from other states, and the system shown in Fig.5.(b) can transit into $IDLE1$ from everywhere (with possibility 1), and the system shown in Fig.5.(c) can transit into $IDLE1$ from state $OFF$ (with possibility 1), $IDLE1$ (with possibility 1), IDLE2 (with possibility 0.5) and $AC$ (with possibility 1), and could not transit from state $HEAT$.

On the other hand, let $\Psi=Ne(\bigcirc IDLE1)$, then $Ne(\square Ne(\bigcirc IDLE1))=Ne(\square \Psi)$. Since $||\Psi||=\neg Po(\bigcirc \neg IDLE1)$ and $Ne(\square \Psi)=\neg Po(\lozenge\neg \Psi))$, using Algorithm 2, by a simple calculation, we have $Ne(\square Ne(\bigcirc IDLE1))=(0,0,0,0)$ for GPKS in Fig.5(a) (b), and $Ne(\square Ne(\bigcirc IDLE1))=(0,0,0,0,0)$ for GPKS in Fig.5(c).
It means that it is unnecessary that the systems shown in Fig.5.(a), (b) and (c) could transit into $IDLE1$ from everywhere.

To sum up the results of Table 1 and Table 2 for Prop.1p, it is unnecessary that the systems shown in Fig.5.(a), (b) and (c) could transit into $IDLE1$ from everywhere. Furthermore, it is not possible that the system shown in Fig.5(a) can transit into $IDLE1$ from states $IDLE2$ and $HEAT$, and it is not possible that the system shown in Fig.5(c) can transit into $IDLE1$ from $HEAT$. It is possible that the system shown in Fig.5.(a) transits into $IDLE1$ from the state $OFF$ (with possibility 1) and $IDLE1$ (with possibility 1), and the system shown in Fig.5.(b) transits into $IDLE1$ from everywhere (with possibility 1), and the system shown in Fig.5.(c) can transit into $IDLE1$ from state $OFF$ (with possibility 1), $IDLE1$ (with possibility 1), IDLE2 (with possibility 0.5) and $AC$ (with possibility 1).


\section{Conclusion}

We introduced possibilistic computation tree logic model checking based on generalized measures, which forms an extension of PoCTL model checking introduced in \cite{li13}. First, the system models were described as generalized possibilistic Kripke structures, and the properties of the systems were specified as generalized computation tree logic formulae. Then the corresponding model checking was discussed, and Algorithm 1-2 was provided to solve the generalized computation tree logic model-checking problems. Next, GPoCTL and PoCTL were compared in detail. Compared with PoCTL, GPoCTL contains more possible and necessary information, even if we use PKS models. The logic GPoCTL is similar to CTL in multi-valued case. Of course, some measure information, including possibility measure and necessity measure, is contained in GPoCTL, whereas there is no measure information in multi-valued CTL model checking. An illustrative example in multi-valued case was used to verify our method.

Further case study needs to be provided. Another direction is the equivalence and abstraction techniques in GPoCTL. For linear-time properties, LTL model checking based on generalized measures using GPKS as system model is another future direction to study (cf.\cite{li12}).

\section*{Acknowledgments}

The authors would like to thank the anonymous referees for helping them refine the ideas
presented in this paper and improve the clarity of the presentation.

\end{document}